\documentclass[leqno]{siamart1116}
\usepackage[T1]{fontenc}
\usepackage[utf8]{inputenc}
\usepackage{amsmath,amsfonts,amssymb,bbold,bbm,tikz}
\usepackage[ruled,linesnumbered]{algorithm2e}
\numberwithin{theorem}{section}

\def\bar{\overline}\def\tilde{\widetilde}
\def\RR{\mathbb R} \def\sig{ \varphi }
\def\b{\boldsymbol}\def\set{\mathit\Omega}

\def\blue#1{{\textcolor{black}{#1}}}

\newenvironment{rev}{\color{black}}{\color{black}}

\title{A Nonlinear Spectral Method for Core--Periphery Detection in Networks\thanks{Submitted to the editors on \today \funding{The work of F.T. is supported by the Marie Curie Individual Fellowship ``MAGNET'' n.\ 744014. The work of D.J.H. is supported by 
grant EP/M00158X/1 from the EPSRC/RCUK Digital Economy Programme.
}}}
\author{Francesco Tudisco\thanks{Department of Mathematics and Statistics, University of Strathclyde, G11XH Glasgow, UK (\email{f.tudisco@strath.ac.uk}), (\email{d.j.higham@strath.ac.uk})} \and Desmond J. Higham\footnotemark[2]}

\begin{document}
	
	\maketitle
	
\begin{abstract}
We derive and analyse a new iterative algorithm for detecting 
network core--periphery structure.
Using techniques in nonlinear Perron-Frobenius theory,
we prove global convergence to the unique solution of a relaxed version of a natural discrete 
optimization problem.
On sparse networks, the cost of each iteration scales linearly with the number of nodes, making the  algorithm feasible for large-scale problems.
We give an alternative interpretation of the algorithm from the perspective of maximum likelihood reordering of a 
new logistic core--periphery  
random graph model. This viewpoint also gives a new basis for quantitatively 
judging a core--periphery detection algorithm.
   We illustrate the algorithm on a range of synthetic and real networks, and show that it 
offers advantages over the current state-of-the-art.
\end{abstract}



\section{Motivation}
\label{sec:mot}
Large, complex networks record pairwise interactions between components in a system.
In many circumstances, we wish to
summarize this wealth of information by 
 extracting high-level information or visualizing key features.
Two of the most important and well-studied tasks are
\begin{itemize}
\item \emph{clustering}, also known as \emph{community detection},
where we attempt to subdivide a network into smaller modules such that
nodes within each module share many connections and nodes in distinct 
modules share few connections, and
\item determination of \emph{centrality} or \emph{rank}, where we assign a nonnegative value to 
each node such that a larger value indicates a higher level of importance.
\end{itemize}
A distinct but closely related problem is to assign each node to either the \emph{core} or
\emph{periphery} in such a way that core nodes are strongly connected across the whole network whereas peripheral nodes are strongly connected only to core nodes; hence there are relatively 
weak 
periphery--periphery connections.
More generally, we may wish to assign a non-negative value to each node, with a larger value indicating greater ``coreness.''
The images in the centre and right of Figure~\ref{fig:sigmoid} indicate the two-by-two block pattern
associated with a core--periphery structure.

The core--periphery concept emerged implicitly in the study of  
economic, social and scientific citation networks, and was formalized in a seminal paper
of Borgatti and Everett \cite{borgatti2000models}. 
A review of recent work on modeling and analyzing  
core--periphery
structure, and 
related ideas in 
degree assortativity, 
 rich-clubs and nested/bow-tie/onion networks, can be found in 
\cite{csermely2013structure}. 
We focus here on the issue of \emph{detection}: given a large complex network with 
nodes appearing in arbitrary order, can we discover, quantify and visualize any inherent  
core--periphery organization? 

In the next section, we set up our notation and discuss background material.
Many detection algorithms can be motivated from an 
optimization perspective. In 
section~\ref{sec:log} we use such an approach to define and justify  
the logistic core--periphery detection problem. We also show how it relates to a new random graph model that generates core--periphery networks.
In section~\ref{sec:nonlinear_PF} we prove that a suitably relaxed version of 
this discrete optimization problem may be solved efficiently using a nonlinear spectral method. 
The resulting algorithm is 
described in subsection~\ref{subsec:alg}. 
Experiments on real and synthetic networks are performed in section~\ref{sec:experiments},
and some conclusions are given in 
section~\ref{sec:disc}.

\section{Background}
\label{sec:bg}
\subsection{Notation}
We use bold letters to denote vectors and capital letters to denote matrices. The respective entries are denoted with lower case, non-bold symbols; 
for example $\b x$ denotes the vector with $i$th entry $x_i$ and $A$ denotes the matrix with $i,j$th entries $a_{ij}$, $i,j=1, \dots, n$.  We use standard entry-wise notation and operations, so for instance $\b x\geq 0$ denotes a vector with nonnegative entries, $|\b x|$ the vector with entries $(|\b x|)_i=|x_i|$, $e^{\b x}$ the vector with entries $(e^{\b x})_i=e^{x_i}$,  
and 
$\b x \b y$ the vector with entries $(\b x\b y)_i=x_iy_i$. For $p \ge 1$ we denote by $\|\b x\|_p = (x_1^p+\cdots +x_n^p)^{1/p}$ the $p$-norm, with $\mathcal S_p = \{\b x:\|\b x\|_p=1\}$ the $p$-unit sphere, and by $\RR_+^n = \{\b x: x_i \geq 0, \forall i\}$ the cone of vectors with nonnegative entries. 

We use $A\in \RR^{n\times n}$  to represent the adjacency matrix of a network $G=(V,E)$,
with vertex set $V$ and edge set $E$. 
We consider undirected networks, so  
$A$ 
is symmetric.
Nonnegative weights are allowed, with a larger value of $a_{ij}$ 
indicating a stronger connection between nodes $i$ and $j$. 
We assume that the network is connected; that is, every pair of nodes may be joined by a path of
edges having nonzero weight.
For a disconnected network 
we could simply 
consider each connected component separately. 

\subsection{Core--periphery Quality Functions}\label{sec:methods}\label{subsec:quality}
Several models for core--periphery detection are based on the definition of various 
\textit{core--periphery quality functions} $f$ and their optimization over certain discrete or continuous sets of vectors. In this setting,
node $i$ is assigned a value $x^\star_i$, where $\b x^\star$ solves an
 optimization problem of the form 
\begin{equation}\label{eq:quality-function}
\max_{\b x \, \in\, \set}\,  f(\b x),\quad f(\b x) = \sum_{i,j=1}^n a_{ij}\, \kappa(x_i,x_j), 	
\end{equation} 
for some choice of kernel function $\kappa$ and constraint set $\set$. 
A larger value of $x^\star_i$ indicates greater ``coreness'', and the 
overall core--periphery structure may be examined  by visualizing the adjacency matrix 
with nodes ordered  
according to the magnitude of the entries of $\b x^\star$.
We mention below some concrete examples.

The influential work of Borgatti and Everett \cite{borgatti2000models} proposed  a discrete notion of core--periphery structure based on comparing 
the given network with 
a block model that consists of a fully connected core and a periphery that has no
internal edges but is fully connected to the core. Their method aims to find an indicator vector
$\b x$ with binary entries. So $x_i = 1$ assigns nodes to the core and 
$x_i=0$ assigns nodes to the periphery. 
By defining the matrix $C=(c_{ij})$ as $c_{ij} = 1$ if $x_i = 1$ or $x_j = 1$ and $c_{ij} = 0$ otherwise, they look at the quantity $
\rho_C = \sum_{ij} a_{ij}c_{ij}
$
and aim to compute the binary vector $\b x$ that maximizes $\rho_C$ among all possible reshufflings of $C$ such that the number of 1 and 0 entries is
preserved.  Clearly this method corresponds to \eqref{eq:quality-function} with $\kappa(x,y) = \mathrm{sign}(x+y)$ and $\set = \{\b x\in \{0,1\}^n : \sum_i x_i = m\}$, for a fixed positive integer $m\leq n$. 

Another popular technique, used for instance in UCINET \cite{ucinet}, is based  on the best rank-one approximation of the off-diagonal entries of $A$. In other words, this method seeks $\b x\in \RR^n$ that  minimizes $\sum_i\sum_{j\neq i}(a_{ij}-x_ix_j)^2$. This is done via the MINRES algorithm, as discussed, for instance, in \cite{minres}.
Writing
 $$ 
A = \lambda_1 \b v_1 \b v_1^T + \lambda_2 \b v_2 \b v_2^T +\cdots + \lambda_n \b v_n \b v_n^T\, ,
 $$
where 
 $\lambda_1>0$ is the largest eigenvalue of $A$ and $\b v_1$ the corresponding eigenvector, it follows that the 
 the optimal rank-one matrix $\b x\b x^T$ we are looking for is strictly related to $\lambda_1 \b v_1 \b v_1^T$. Therefore the least-squares problem is equivalent to maximizing the Rayleigh quotient of $A$; that is,  the following optimization problem
 \begin{equation}
\max_{\b x \neq 0}\frac{\b x^T A \b x}{\b x^T \b x}.
\label{eq:RQ} 
\end{equation}
This, in turn, coincides with \eqref{eq:quality-function} for 
$\kappa(x,y) = xy$ and $\set = \{\b x : \b x^T \b x=1\}=\mathcal S_2$. Moreover, as the matrix $A$ is symmetric, nonnegative and irreducible, by the Perron-Frobenius theorem, the maximizer  $\b v_1$ is unique and entrywise positive and the corresponding eigenvalue $\lambda_1$ coincides with the spectral radius of $A$. Following a different construction, the use of the spectral radius and the associated Perron eigenvector of $A$ for detecting core--periphery is also considered in \cite{mondragon2016network}. Note that, thanks to the Perron-Frobenius theorem, it follows that the constraint set in \eqref{eq:quality-function} can be chosen as $\set = \mathcal S_2^+=\mathcal S_2\cap \RR^n_+$. This observation has practical importance 
 because it constrains the solution space. As we discuss in Section~\ref{sec:nonlinear_PF}, 
this feature is shared by our nonlinear core--periphery model, where existence and uniqueness are 
proved using a customized nonlinear Perron-Frobenius-type theorem. Moreover, note that having a nonnegative solution $\b x$ to \eqref{eq:quality-function} not only allows for a core--periphery assignment or ranking, but also implicitly produces a continuous core--periphery score for the nodes.
We note that the Perron--Frobenius eigenvector of $A$ is also a well-known nodal centrality measure
\cite{EH10}.

The concept of core--periphery quality measure with general kernel function, as formulated in \eqref{eq:quality-function},  was introduced by Rombach et al.\ in \cite{rombach2014core}. Those authors focus on the choice $\kappa(x,y)=xy$ and introduce a novel continuous constraint set defined in terms of two parameters $0\leq \alpha, \beta \leq 1$ as follows
\begin{equation}\label{eq:C_alpha_beta}
	\set = C_{\alpha, \beta} = \left\{\b x \in \RR^n : \begin{array}{ll}x_i = \frac{i(1-\alpha)}{2\lfloor\beta n\rfloor} & \text{for } i=1,\dots,\lfloor\beta n\rfloor,\\  x_i = \frac{(i-\lfloor\beta n\rfloor)(1-\alpha)}{2(n-\lfloor\beta n\rfloor)}+\frac{1+\alpha}{2} &\text{for } i = \lfloor\beta n\rfloor +1, \dots, n \end{array}\right\}.
\end{equation}
Here $\alpha$ is used to tune the score jump between the peripheral node with highest score and the core node with lowest score, whereas $\beta$ is used to set the size of the core set. Note that,  as $0\leq \alpha, \beta \leq 1$, we have $C_{\alpha,\beta}\subseteq \RR^n_+$ and thus, as for the Perron--Frobenius eigenvector of $A$, the maximizer of \eqref{eq:quality-function} with $\kappa(x,y)=xy$ and $\set = C_{\alpha,\beta}$  is a nonnegative vector whose entries define a core--periphery score value, called the \textit{aggregate core score} in \cite{rombach2017core}. 

\subsection{The Optimization Problem}\label{ssubsec:opt}
The models proposed in \cite{borgatti2000models} and 
\cite{rombach2014core} lead to
discrete optimization problems whose global solution cannot be computed for large graphs. 
Both papers propose computational methods that deliver 
approximate solutions
but do not come with guarantees of accuracy. The combinatorial optimization problem of \cite{borgatti2000models} is solved via random reshuffling. For the model proposed in \cite{rombach2014core}
a simulated--annealing algorithm is used.
The presence of the two parameters, $\alpha$ and $\beta$, 
adds a complication, which is addressed there by considering all $(\alpha,\beta)$ 
values on a discrete uniform lattice in $[0,1]^2$.  Clearly, refining the discretization level 
improves the approximation to the solution but  raises the computational cost.

For the model used in UCINET based on MINRES
\cite{ucinet,minres},
recalling 
(\ref{eq:RQ}) we note that 
an efficient approach is to 
recast the optimization problem into the computation of a matrix eigenvector, for which well-established algorithms are available.

Since our approach fits into the core--periphery quality function optimization approach
of 
\cite{rombach2014core,rombach2017core}, we will use the method developed there, 
with $\kappa(x,y)=xy$ and $\Omega=C_{\alpha,\beta}$, as a baseline for  comparison in our experiments in Section~\ref{sec:experiments}. 

Although algorithms based on other choices of the kernel function $\kappa$ have not been considered in the literature so far, both in Section 2.2.1 of \cite{rombach2014core} and Section 4.2.1 of \cite{rombach2017core} it is pointed out that an ideal core--periphery kernel function
is 
\begin{equation}\label{eq:mu_alpha}
 	\kappa(x,y)=\mu_\alpha(x,y)=(|x|^\alpha+|y|^\alpha)^{1/\alpha}
 \end{equation} 
for $\alpha>0$ large. In fact this function is related to core--periphery structure in a very natural way, as we discuss in the next section. 

\section{Logistic Core--Periphery Detection Problem}\label{sec:log}
We propose a new model based on the kernel $\kappa(x,y) = \max\{|x|,|y|\}$.
Note that this kernel function arises as the $\alpha \to \infty$ limit of \eqref{eq:mu_alpha}.
Focusing for now on the ranking problem, our goal is to 
determine a core--periphery ranking vector that assigns to each node a distinct integer 
between $1$ and $n$; with a lower rank denoting a more peripheral  node. Clearly any such  ranking vector is nothing but a permutation vector $\b \pi$, where  $i\mapsto \pi_i$ is a permutation of the set $\{1,\dots,n\}$. Therefore,  if $\mathcal P_n$ is the set of permutation vectors of $n$ entries, we formulate our core--periphery detection problem as follows
\begin{equation}\label{eq:our_model}
\max_{\b \pi \, \in \, \mathcal P_n}f_{\infty}(\b \pi), \quad \text{where}\quad f_\infty(\b \pi) = \sum_{i,j=1}^n a_{ij}\max\{\pi_i,\pi_j\}. 
\end{equation} 

 We will see in Section \ref{sec:nonlinear_PF} that, in practice, finite but large enough values of $\alpha$ 
in \eqref{eq:mu_alpha} 
provide an accurate approximation of $\max\{x,y\}$.
Moreover, relaxing from $\mathcal P_n$ to 
 $\mathcal S_p^+ = \mathcal S_p \cap \RR^n_+$ 
 allows for a globally convergent, easily implementable and computationally 
feasible algorithm.

We will refer to \eqref{eq:our_model} as the \textit{logistic core--periphery detection problem}.
In order to motivate this name and the model itself,  we discuss in the next section
a random graph model that provides a natural and flexible
 model for core--periphery structure. 

\subsection{Logistic Core--periphery Random Graph Model}\label{sec:logistic_random_model}
We now consider random graph models 
that generate core--periphery structure.
For this subsection only, we restrict to the case of unweighted, or binary, networks.
We focus on models where the nodes can be placed in a natural ordering, represented by a permutation vector, so $i\mapsto \pi_i$. In this natural ordering, for every pair of nodes
$i$ and $j$ the 
probability of an edge will be a function 
of $\pi_i$ and $\pi_j$. Moreover, these events will be independent.
We note that such models have been studied in other contexts; for example, in an early reference Grindrod 
\cite{Grin02} 
used this framework to define a class of range-dependent graphs 
that captures features of the classic Watts-Strogatz model.

A simple core--periphery model of this type arises when 
edges are present with probability one within the core and between core and periphery, and with probability zero among peripheral nodes. This model is considered for instance in \cite{borgatti2000models,rombach2014core}.  
In this model  there exists a permutation of the indices $i\mapsto \pi_i$ such that an edge connecting two different nodes $i$ and $j$ exists with independent probability 
$
\mathbb P(i\sim j) = H_t\left(\frac 1 n\max\{\pi_i,\pi_j\}\right), 
$
where, for $t\in (0,1)$, $H_t$ is the Heaviside function $H_t(x)=1$ if $x\geq t$ and $H_t(x)=0$ otherwise. The parameter $t$ allows us to tune the size of the core and of the periphery. 
Figure~\ref{fig:sigmoid} (center) shows an example matrix  whose $ij$-th entry is the probability $\mathbb P(i\sim j)$ from this model, for $t=1/2$ and $\pi_i=n-i$ for any $i$. 
\begin{figure}[t]
\includegraphics[width=.26 \textwidth]{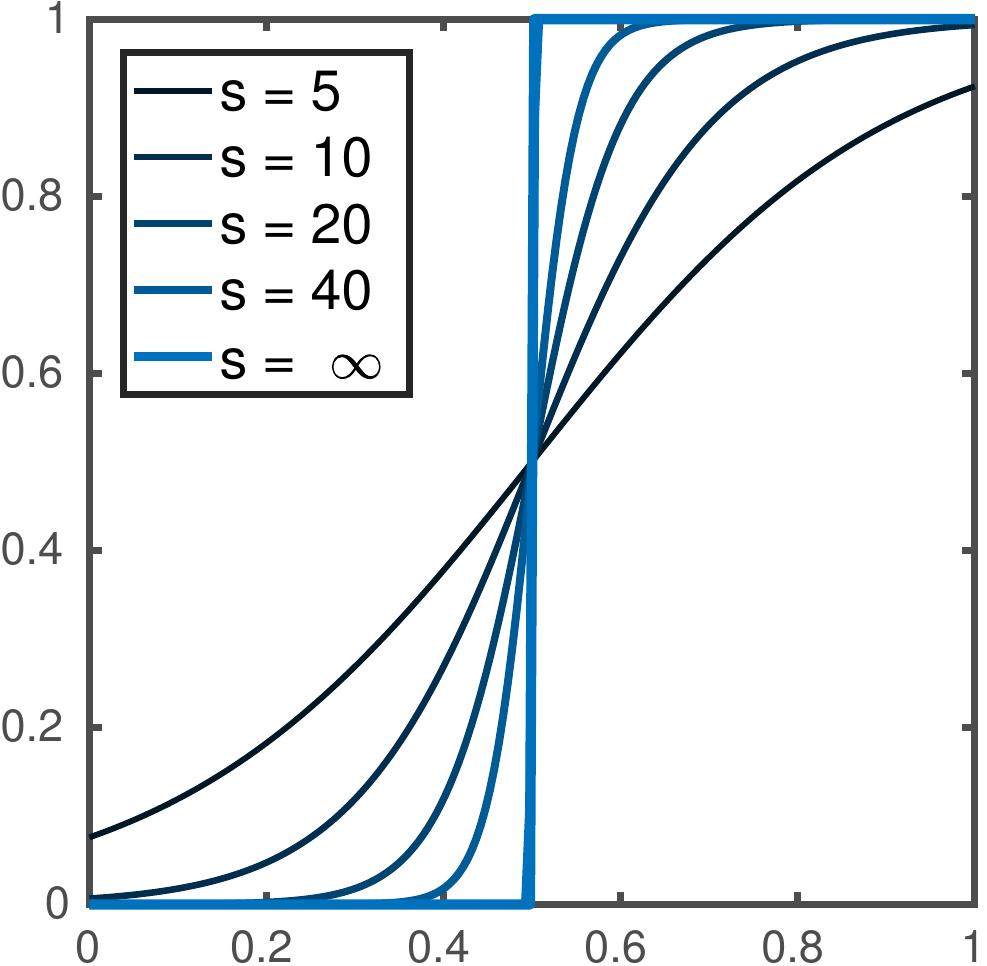}\hfill
\includegraphics[width=.25\textwidth]{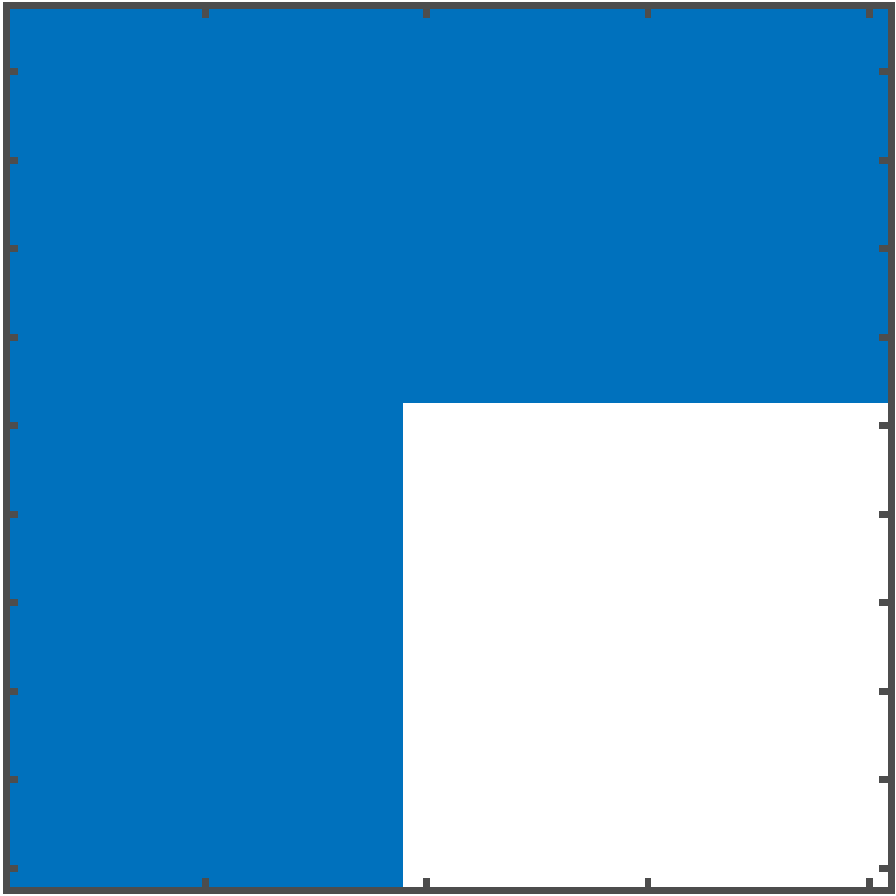}\hfill
\includegraphics[width=.305\textwidth]{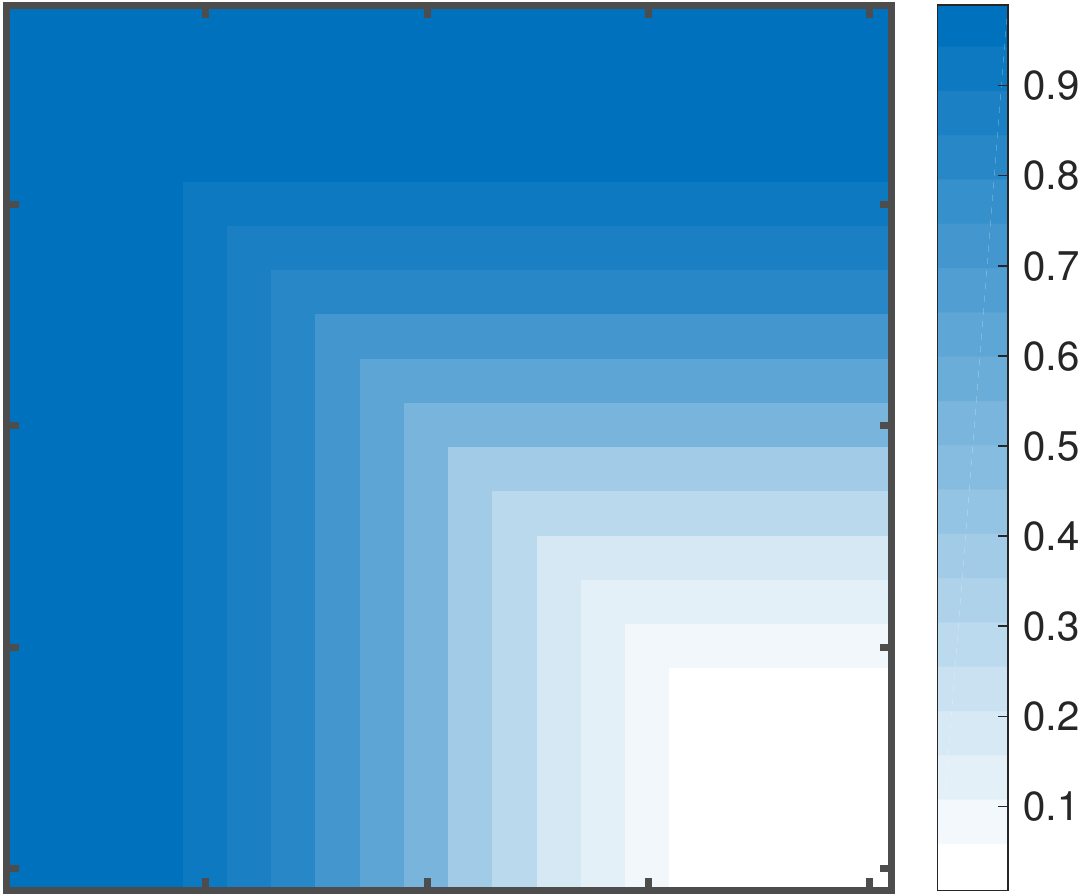}
\caption{Left: $\sigma_{s,t}(x)$ for $t=1/2$ and $s \in \{5, 10, 20, 40\}$. The
piecewise linear plot is the Heaviside function $H_{t}(x)$ which corresponds to $\lim_{s\to\infty}\sigma_{s,t}(x)$. Center: heatmap of $A$ with entries $a_{ij} = H_{1/2}(\max\{1-i/n, 1-j/n\})$. Right: heatmap of $A$ with entries $a_{ij}=\sigma_{10,1/2}(\max\{1-i/n, 1-j/n\})$. In these heatmaps, in order to emphasize the overall structure,
 the diagonal entries have been colored with the probabilities $\mathbb P(i\sim i)=H_{1/2}(1-i/n)$ and $\mathbb P(i\sim i)=\sigma_{10,1/2}(1-i/n)$, respectively. However, the associated random graphs have no self-loops and  so
the actual diagonal probabilities are $\mathbb P(i\sim i) = 0$.}\label{fig:sigmoid}
\end{figure}
The Heaviside function $H_t$ is a discontinuous step function, and 
it leads to a idealized all-or-nothing structure. 
Instead, we may consider a family of continuous approximations to $H_t$ based on the logistic sigmoid function. 
For $s,t\in \RR$, $s\geq 0$ we define
$$
\sigma_{s,t}(x) = \frac{1}{1+e^{-s(x-t)}} \, .
$$
Note that, for any fixed $x,t\in \RR$ we have $\lim_{s\to\infty}\sigma_{s,t}(x)=H_t(x)$. 
Examples are plotted in Figure~\ref{fig:sigmoid} (left). 

We now introduce the  random graph model where an edge connecting two different nodes $i$ and $j$ exists with independent probability
$$
\textstyle{\mathbb P(i\sim j) = \sigma_{s,t}\left(\frac 1 n \max\{n-i,n-j\}\right) \, .}
$$ 
We refer to this as the \emph{logistic core--periphery random graph model}. 
The right-most plot on Figure~\ref{fig:sigmoid} shows a $20\times 20$ example matrix whose $ij$-th entry 
is the corresponding probability $\mathbb P(i\sim j)$, for $s=10$ and $t=1/2$. 
We see that, relative to the Heaviside version, this model gives a smoother transition from 
core to periphery, and has a built-in notion of ranking within each group.  The relevance of this model to capture core and perhipheral nodes has been also recently pointed out in \cite{jia2018detecting}.

\begin{rev}We are interested in the circumstance where a core--periphery structure is present in the graph, but must be discovered. In practice, 
our task is to find a suitable reordering of the nodes that highlights the presence of core and periphery.  
A natural approach is then to 
find the permutation of indices $\b \pi \in \mathcal P_n$ that maximizes the 
likelihood, under the assumption of a logistic core--periphery structure.\end{rev}
This likelihood is given by
\begin{equation}\label{eq:LCP_probability}
\nu(\b \pi)=\prod_{i\sim j}\sig(\pi_i,\pi_j)\, \prod_{i\not\sim j}\left(1-\sig(\pi_i,\pi_j)\right), 
\end{equation} 
where,  for the sake of brevity, we let $\sig(x,y)=\sigma_{s,t}(\frac 1 n \max\{x,y\})$.
We now show that solving the proposed logistic core--periphery detection problem \eqref{eq:our_model} 
is equivalent to solving this maximum likelihood reordering problem.

\begin{theorem} \label{thm:mlp}
	$\b \pi^\star\in \mathcal P_n$ is a permutation that maximizes $\nu(\b \pi)$ if and only if $\b \pi^\star$ is a solution of \eqref{eq:our_model}.
\end{theorem}
\begin{proof}
  Our proof exploits a very useful trick that Grindrod 
\cite{Grin02}
used in the case of a range-dependent random graph:
 The likelihood  $\nu(\b \pi)$ can be equivalently written as 
	\begin{equation*}\label{eq:nu_2}
	\nu(\b \pi) = \prod_{i\sim j} \frac{\sig(\pi_i,\pi_j)}{1-\sig(\pi_i,\pi_j)}\, \prod_{i,j=1}^n 1-\sig(\pi_i,\pi_j)\, .
	\end{equation*}
	As the right-hand factor does not depend on the graph, maximizing $\nu(\b \pi)$ is equivalent to maximizing the left-hand factor. Thus, taking the logarithm on both sides we observe that $\b \pi^\star$ maximizes $\nu$ if and only if it maximizes 
	$$\sum_{ij=1}^n a_{ij}\log\left(\frac{\sig(\pi_i,\pi_j)}{1-\sig(\pi_i,\pi_j)}\right)\,.$$  Now, using the definition of $\sig(x,y)$ in terms of the logistic sigmoid function, a short computation shows that $\log(\sig(x,y)/(1-\sig(x,y))=s(\frac 1 n \max\{x,y\}-t)$, for any $x,y,t\in \RR$, $s\geq 0$. Therefore, $\b \pi^\star$ maximizes $\nu$ if and only if it maximizes the core--periphery quality function $\sum_{ij}a_{ij}\max\{\pi_i,\pi_j\}$, which concludes the proof.
\end{proof}

In words, Theorem~\ref{thm:mlp} shows that in the case of unweighted networks, solving the 
logistic core--periphery detection problem
\eqref{eq:our_model} 
is equivalent to solving the maximum likelihood reordering problem
(\ref{eq:LCP_probability}) under the assumption that the network was generated from
the logistic core--periphery random graph model. \begin{rev} This is somewhat analogous to a known phenomenon in the community detection case \cite{newman2016equivalence}. \end{rev}

We mention that  core--periphery detection via
likelihood maximization on a random graph model was also proposed in 
\cite{Zhang15}.  
There, the authors used a stochastic block model where nodes are independently
assigned to the core with probability $\gamma_1$ and to the periphery with probability
$1- \gamma_1$. 
 Core--core, core--periphery and periphery--periphery connections then appear with
independent probabilities $p_{11}$, $p_{12}$ and $p_{22}$, with 
$p_{11} > p_{12} > p_{22}$. 
Infering model parameters  by maximizing the likelihood over all possible node bi--partitions
leads to a core--periphery assignment.  
Because solving this discrete optimization problem is not practicable for large networks, the 
authors develop an approximation technique based on expectation maximization and 
belief propagation.
We emphasize that this random graph reordering/partitioning framework applies to unweighted 
(binary) networks.

\section{Nonlinear Spectral Method for Core--periphery Detection}\label{sec:nonlinear_PF}
 In this section we introduce an iterative method for the logistic core--periphery detection
problem \eqref{eq:our_model}
and prove that it converges globally to the solution of a relaxed problem. We 
refer to this as a \textit{nonlinear spectral method} for two reasons.
First, its derivation and analysis 
are inspired by recent work in nonlinear Perron-Frobenius theory \cite{gautier2016globally,mhpfpaper,tudisco2017node,gautier2018contractivity}. Second, 
as shown in Lemma~\ref{lem:eigenvalue_problem}, there is an 
equivalence between \eqref{eq:our_model} and a nonlinear eigenvalue problem. 
Recall that the network is assumed to be (nonnegatively) weighted, connected and undirected.

The logistic core--periphery model \eqref{eq:our_model} is a combinatorial optimization problem whose exact solution is not feasible for large scale networks. We therefore introduce two relaxations that lead to a new ``smooth'' logistic core--periphery problem whose solution may be computed efficiently with a new nonlinear spectral method.

Given $\alpha>1$, we replace the discontinuous kernel function $\max\{|x|,|y|\}$ with
\begin{equation}
\mu_\alpha(x,y) = (|x|^\alpha + |y|^\alpha)^{1/\alpha}.
\label{eq:mualpha}
\end{equation}
 As mentioned 
at the end of subsection~\ref{subsec:quality},
$\max\{|x|,|y|\}$ is the limit of $\mu_\alpha(x,y)$ for $\alpha\to\infty$. More precisely, 
letting
\begin{equation}
 f_\alpha(\b x) = \sum_{ij}a_{ij}\mu_\alpha(x_i,x_j), 
\label{eq:falpha}
\end{equation}
a simple computation using the H\"older inequality reveals that 
\begin{equation}\label{eq:alpha_approx} 
f_\infty(\b x)\leq f_\alpha(\b x)\leq 2^{1/\alpha} f_\infty(\b x)\, ,
\end{equation}
for any $\alpha>1$. Therefore when $\alpha$ is large enough, using $f_\alpha$ in place of $f_\infty$ in \eqref{eq:our_model} provides a very accurate approximation. 

Second, we relax the discrete constraint set $\mathcal P_n$ into a continuous one. In doing this we note that every vector in $\mathcal P_n$ is entry-wise nonnegative and has fixed length. For instance, $\|\b x\|_1 = \frac 1 2 n (n+1)$, for any $\b x \in \mathcal P_n$. 
Note that the normalization constant $\frac 1 2 n (n+1)$ can be chosen arbitrarily. 
In fact, the function $f_\alpha$ we are considering is positively $1$-homogeneous; that is, for any $\lambda>0$ we have $f_\alpha(\lambda \b x)  = \lambda f_\alpha(\b x)$. This implies that if $\b x$ maximizes $f_\alpha$ among all the vectors of norm exactly $1$ then, for any $a>0$, $a\b x$ maximizes $f_\alpha$ among all the vectors of norm exactly $a$. We therefore relax $\mathcal P_n$ into a sphere of nonnegative vectors. For convenience, we choose the $p$-sphere $\mathcal S_p=\{\b x : \|\b x\|_p=1\}$ and let $\mathcal S_p^+ = \mathcal S_p \cap \RR^n_+$.

Overall, for $\mu_\alpha(x,y)$ in
(\ref{eq:mualpha})  and 
$f_\alpha(\b x) $ in (\ref{eq:falpha}), 
we modify the original logistic core--periphery problem \eqref{eq:our_model} into 
\begin{equation}\label{eq:maximization_problem}
\max_{\b x\in \mathcal S_p^+} f_\alpha(\b x), \quad\text{where}\quad f_\alpha(\b x)=\sum_{i,j=1}^n a_{ij}\, \mu_\alpha(x_i,x_j) \, .
\end{equation}

We devote the remainder of  this section to proving that, for any $\alpha>0$ and any $p>\alpha$, the relaxed logistic core--periphery model (\ref{eq:maximization_problem}) has a unique,
entry-wise positive
 solution that can be  efficiently computed via a globally convergent iterative method.

\def\epsilon{\varepsilon}
\subsection{Existence and Uniqueness of a Solution to the Relaxed Problem}
We begin by observing that  the function $f_\alpha$ attains its maximum on a positive vector. 
\begin{lemma}\label{lem:the_solution_is_positive}
The problem \eqref{eq:maximization_problem} is solved by a vector 
$\b x^\star$ such that $\b x^\star>0$.
\end{lemma}
\begin{proof}
 As  $f_\alpha(\b x) = f_\alpha(|\b x|)$ for any $\b x\in \RR^n$, we easily deduce that the maximum is attained on a vector $\b x^\star\geq 0$. Now suppose that there exists $1\leq k\leq n$ such that $x^\star_k=0$. As the graph is connected, there exists $\ell$ such that $a_{k\ell}>0$. Then the vector $\b y$ defined by $y_i = x^\star_i$ for $i\neq k$ and $y_k=\epsilon >0$ would be such that 
\begin{align*}
f_\alpha(\b y) &= \sum_{i,j\neq k}a_{ij}\mu_\alpha(x^\star_i,x^\star_j)+2\sum_{j}a_{kj}\mu_\alpha(x_j^\star,\epsilon)\\
 &\geq \sum_{i,j\neq k}a_{ij}\mu_\alpha(x^\star_i,x^\star_j)+2a_{k\ell}\mu_\alpha(x_\ell^\star,\epsilon) >f_\alpha(\b x^\star), 
\end{align*}
which contradicts the maximality of $\b x^\star$. We conclude that the solution
of  \eqref{eq:maximization_problem} is attained on an entry-wise positive vector. 
\end{proof}

  Now, by using the positive $1$-homogeneity of $f_\alpha$, we show that the constrained optimization problem \eqref{eq:maximization_problem} is equivalent to an unconstrained problem for the normalized function $f_\alpha(\b x)/\|\b x\|_p$. 
\begin{lemma}\label{lem:balls_equivalence}
For any $p>1$ and any $\alpha>1$ we have
$$
\max_{\b x\in \mathcal S_p^+}f_\alpha(\b x) = \max_{\b x\in\RR^n}\frac{f_\alpha(\b x)}{\|\b x\|_p}. 
$$
\end{lemma}
\begin{proof}
By the $1$-homogeneity of $f_\alpha$ we have the following chain of inequalities
\begin{align*}
\max_{\|\b x\|_p\leq 1 }f_\alpha(\b x) &\geq \max_{\|\b x\|_p=1}f_\alpha(\b x) = \max_{\b x\in\RR^n}\,f_\alpha\!\left(\frac{\b x}{\|\b x\|_p}\right) =\max_{\b x\in\RR^n}\frac{f_\alpha(\b x)}{\|\b x\|_p}\\
& \geq \max_{\|\b x\|_p\leq 1 } \frac{f_\alpha(\b x)}{\|\b x\|_p} \geq \max_{\|\b x\|_p\leq 1}f_\alpha(\b x)\, .
\end{align*} 
This implies that the inequalities above are all identities. Together with $f_\alpha(\b x)=f_\alpha(|\b x|)$ this  shows the claim. 
\end{proof} 
We have the following consequence. 
\begin{lemma}\label{lem:eigenvalue_problem}
Let $F_\alpha=\nabla f_\alpha:\RR^n\to\RR^n$ be the gradient of $f_\alpha$, that is,  
$$
F_\alpha(\b x)_i = 2\, \sum_{j=1}^na_{ij}\, |x_i|^{\alpha-2}x_i(|x_i|^\alpha+|x_j|^\alpha)^{1/\alpha -1}\, , \quad i=1,\dots,n\, .
$$
Then, for any $p>1$, the following statements are equivalent:
\begin{enumerate}
\item $\b x$ is a solution of \eqref{eq:maximization_problem}, 
\item $\b x$ satisfies the eigenvalue equation $F_\alpha(\b x) = \lambda\,  |\b x|^{p-2}\b x$ with $\lambda >0$,
\item $\b x$ is a fixed point of the map $G_\alpha(\b x)= |F_\alpha(\b x)|^{q-2}F_\alpha(\b x)/\|F_\alpha(\b x)\|_q^{q-1}$, where $q$ is the H\"older conjugate of $p$, i.e. $1/p+1/q=1$. 
\end{enumerate}
\end{lemma}
\begin{proof}
For convenience, let us write  $r_\alpha(\b x)=f_\alpha(\b x)/\|\b x\|_p$.  By differentiating $r_\alpha(\b x)$ we see that 
$$
\nabla r_\alpha(\b x) =0 \iff \frac{1}{\|\b x\|_p}\Big\{ F_\alpha(\b x) - \lambda \, |\b x|^{p-2}\b x\Big\}=0,
$$
with $\lambda = f_\alpha(\b  x)/\|\b x\|_p^p>0$. Together with Lemma \ref{lem:balls_equivalence} this proves $(1)\iff (2)$. Now note that the map $\b x\mapsto \psi(\b x)=|\b x|^{q-2}\b x$ is such that $\psi(|\b x|^{p-2}\b x)=\b x$. In fact 
$$
\psi(|\b x|^{p-2}\b x)=|\, |\b x|^{p-2}\b x|^{q-2}\, |\b x|^{p-2}\b x = |\b x|^{(p-1)(q-2)+p-2}\b x=\b x\, .
$$
As $\psi$ is bijective we have $F_\alpha(\b x)=\lambda\,|\b x|^{p-2}\b x\iff |F_\alpha(\b x)|^{q-2}F_{\alpha}(\b x) =\psi(\lambda)\b x$. Therefore, recalling that $\|\b x\|_p=1$ and $\lambda>0$ we have  
\begin{align*}
F_\alpha(\b x)=\lambda\,|\b x|^{p-2}\b x \iff \frac{|F_\alpha(\b x)|^{q-2}F_{\alpha}(\b x)}{\|F_\alpha(\b x)\|_q^{q-1}}=\frac{|F_\alpha(\b x)|^{q-2}F_{\alpha}(\b x)}{\||F_\alpha(\b x)|^{q-2}F_{\alpha}(\b x)\|_p}=\frac{\psi(\lambda)\b x}{\|\psi(\lambda)\b x\|_p}=\b x, 
\end{align*}
where the first identity follows by $\||\b y|^{q-2}\b y\|_p = \|\b y^{q-1}\|_p =\|\b y\|_q^{q-1}$. This shows $(2)\iff(3)$ and concludes the proof.
\end{proof}

 We need one final rather technical lemma that, for the sake of completeness, we state for the  case where $\alpha$ may attain both positive and negative values.
\begin{lemma}\label{lem:lip_constant}
For $\alpha\in\RR$ let $G_\alpha$ be defined as in Lemma \ref{lem:eigenvalue_problem} above, and let $g_i:\RR^n\to\RR$ be the scalar functions such that  $G_\alpha(\b x)=(g_1(\b x), \dots, g_n(\b x))$. Then
$$\left\| \frac{\nabla g_k(\b x)\b x}{g_k(\b x)}\right\|_1 \leq \frac{|1-\alpha|}{p-1}$$
for any vector $\b x\geq 0$ and any $k=1,\dots,n$. 
\end{lemma}
\begin{proof}
	First note that $\b y\geq 0$ implies that both $F_\alpha(\b y)$ and $G_\alpha(\b y)$ are nonnegative. Now, let $i,k\in \{1,\dots,n\}$. By the definition of $g_k$ and using the chain rule, we obtain 
	\begin{align}
	\begin{aligned}\label{eq:derivative}
	\left|\frac{\partial_i \{g_k(\b y)\}}{g_k(\b y)}\right| &= \left|\frac{\partial_i \Big\{\big(F_\alpha(\b y)_k\big)^{q-1}\Big\}}{\,\,\,\big(F_\alpha(\b y)_k\big)^{q-1}} - \frac{\partial_i\Big\{ \|F_\alpha(\b y)\|_q^{q-1}\Big\}}{\,\,\,\, \|F_\alpha(\b y)\|_q^{q-1}}\right|\\
	&= (q-1) \left| \frac{\partial_i F_\alpha(\b y)_k}{F_\alpha(\b y)_k}- \frac{\sum_m \big(F_\alpha(\b y)_m\big)^{q-1}\partial_i F_\alpha(\b y)_m}{\|F_\alpha(\b y)\|_q^q}\right|\\
	&= (q-1) \left| \frac{\partial_i F_\alpha(\b y)_k}{F_\alpha(\b y)_k}- \sum_m\left(\frac{F_\alpha(\b y)_m}{\|F_\alpha(\b y)\|_q}\right)^q \frac{\partial_i F_\alpha(\b y)_m}{F_\alpha(\b y)_m}\right|\, .
	\end{aligned}
	\end{align}
Now note that for any  nonnegative $\b y\in\RR^n$ we have 
$$
\frac{\partial_i F_\alpha(\b y)_m y_i}{F_\alpha(\b y)_m} = (1-\alpha)\,\frac{y_i^\alpha}{y_i^\alpha+y_m^\alpha}\,\frac{a_{mi}(y_m^\alpha+y_i^\alpha)^{1/\alpha-1}} {\sum_j a_{mj}(y_m^\alpha+y_j^\alpha)^{1/\alpha-1}} \geq 0 \, .
$$
Let $\tilde m\in\{1,\dots, n\}$ be the index for which the quantity above is maximal. From \eqref{eq:derivative} we deduce that $|\partial_i \{g_k(\b y)\}y_i/g_k(\b y)|$ is of the form $|\beta_k - \sum_m \gamma_m \beta_m|$, where $\beta_m, \gamma_m\geq 0$ for all $m=1,\dots,n$ and $\sum_m \gamma_m=1$. This implies that $|\beta_k - \sum_m \gamma_m \beta_m|\leq \beta_{\tilde m}$ and, as $y_i^\alpha/(y_i^\alpha+y_{\tilde m}^\alpha)\leq 1$, we find 
$$
\left|\frac{\partial_i \{g_k(\b y)\}y_i}{g_k(\b y)}\right| \leq (q-1)|1-\alpha|\left( \frac{a_{\tilde m,i}(y_{\tilde m}^\alpha+y_i^\alpha)^{1/\alpha-1}} {\sum_j a_{\tilde m,j}(y_{\tilde m}^\alpha+y_j^\alpha)^{1/\alpha-1}} \right)\, .
$$
Summing this formula over $i$ and recalling that $q-1=1/(p-1)$ concludes the proof.
\end{proof}
This leads to our main result.
\begin{theorem}\label{thm:convergence}
Assume $\alpha >1$ and $p>\alpha$. Then \eqref{eq:maximization_problem} has a unique solution $\b x^\star$ that is entry-wise positive. Moreover, for any $\b x_0>0$, if $G_\alpha$ is defined as in Lemma \ref{lem:eigenvalue_problem}, then   the sequence  defined by $\b x_{k+1}=G_\alpha(\b x_k)$ belongs to $\mathcal S_p^+$ and converges to  $\b x^\star$. 
\end{theorem}
\begin{proof}
The fact that $\b x_k\in \mathcal S_p^+$ for any $k$ is an obvious consequence of the identities $\||\b z|^{q-2}\b z\|_p = \|\b z^{q-1}\|_p =\|\b z\|_q^{q-1}$. Now we show that the map $G_\alpha$ defined in Lemma \ref{lem:eigenvalue_problem} is Lipschitz contractive which, 
due to the Banach fixed point theorem, 
gives convergence of the sequence and uniqueness of the solution.
To this end we use the Thomson metric $d_T$ defined for $\b x,\b y\in \mathcal S_p^+$ as $d_T(\b x,\b y) = \|\log(\b x)-\log(\b y)\|_\infty$.  

As before, for $i=1,\dots,n$,  let $g_i:\RR^n\to\RR$ be the scalar functions such that $G_\alpha(\b x)=(g_1(\b x), \dots, g_n(\b x))$. By the Mean Value Theorem we have 
$$
\phi(\b x)-\phi(\b y)=\nabla \phi(\b \xi)^T(\b x-\b y)
$$
for any differentiable function $\phi:\RR^n \to\RR$ and with $\b \xi$ being a point in the segment joining $\b x$ and $\b y$. Consider the function $\phi(\b x) = \log(g_i(e^{\b x}))$. Then $\nabla \phi(\b \xi) = \nabla g_i(e^{\b \xi})e^{\b \xi}/g_i(e^{\b \xi})$ and we obtain 
\begin{equation*}
|\phi(\b x)-\phi(\b y)|=|\log(g_i(e^{\b x}))-\log(g_i(e^{\b y}))| = \left| \frac{(\nabla g_i(e^{\b \xi})e^{\b \xi})^T(\b x-\b y)}{g_i(e^{\b \xi})}  \right|. 
\end{equation*}
As the exponential function maps positive vectors into positive vectors bijectively, the previous equation implies that for any two positive vectors $\bar {\b x} = e^{\b x}$ and $\bar {\b y} = e^{\b y}$ we have 
$$ |\log(g_i(\bar {\b x}))-\log(g_i(\bar {\b y}))| \leq \left\| \frac{\nabla g_i(e^{\b \xi})e^{\b \xi}}{g_i(e^{\b \xi})}\right\|_1  \!\|\log(\bar {\b x})-\log(\bar {\b y})\|_\infty \!= \left\| \frac{\nabla g_i(e^{\b \xi})e^{\b \xi}}{g_i(e^{\b \xi})}\right\|_1 \!d_T(\bar {\b x}, \bar {\b y}).
$$
Together with Lemma \ref{lem:lip_constant} we have 
$d_T(G_\alpha(\b x),G_\alpha(\b y))\leq C d_T(\b x,\b y)$ with $C=|1-\alpha|/(p-1)<1$. Thus $G_\alpha$ is a contraction and $\b x_k\to \b x^\star\in \mathcal S_p^+$ as $k\to \infty$. Finally, Lemmas \ref{lem:the_solution_is_positive} and \ref{lem:eigenvalue_problem} imply that 
$\b x^\star$ is entry-wise positive and solves \eqref{eq:maximization_problem}, concluding the proof.
\end{proof}

\subsection{Algorithm}\label{subsec:alg}
Theorem~\ref{thm:convergence} leads naturally to the following algorithm. 

\SetKwBlock{Repeat}{For $k=0,1,2,3,\dots$ repeat}{{until $\|\b x_{k}-\b x_{k+1}\|/\|\b x_{k+1}\|< \mathrm{tolerance}$}}
\begin{algorithm}[H]
		\caption{Nonlinear spectral method for core--periphery detection}\label{alg:1}
		 \DontPrintSemicolon
		\KwIn{Adjacency matrix $A$, initial guess $\b x_0 >0$\\
		\hspace{3.6em}Fix $\alpha \gg 1$, $p>\alpha$, $q=p/(p-1)$\\
		\hspace{3.6em}Let $F_\alpha(\b x)_i = 2\, \sum_{j=1}^na_{ij}\, |x_i|^{\alpha-2}x_i(|x_i|^\alpha+|x_j|^\alpha)^{1/\alpha -1}$, for $i=1, \dots, n$}
		\Repeat{
		$\b y_{k+1} = F_\alpha(\b x_k)$ \;
		$\b x_{k+1} = \|\b y_{k+1}\|_q^{1-q}|\b y_{k+1}|^{q-2}\b y_{k+1}$\;
		}
		$\b c = \b x_{k+1}/\max(\b x_{k+1})$\;
		Reorder the network nodes according to the magnitude of the entries of $\b c$\;
		\KwOut{Core--periphery score $\b c$ and approximate maximizer $\b x_{k+1}$ of \eqref{eq:maximization_problem}.}
\end{algorithm}
Recall that, since $\b x_0>0$, each element of the sequence $\b x_k$ generated by the algorithm is a positive vector, due to Theorem \ref{thm:convergence}. Each iteration requires the computation of a vector norm, at step 3, and the computation of the action of the nonlinear map $F_\alpha$ on a nonnegative vector $\b x$, at step 2. Thus, if $m\geq n$ is the number of edges in the network (or equivalently half the number of nonzero entries of $A$), the order of complexity per iteration of Algorithm~\ref{alg:1} is $O(m)+O(n)$. 
For large-scale, sparse, real-world networks $m$ is typically linearly proportional to $n$ or $n\log n$. \begin{rev}In this setting the method is scalable to high dimensions, as  confirmed by Figure \ref{fig:NSM_timing}. \end{rev}
\begin{figure}[t]
\centering \includegraphics[width=.9\textwidth,clip,trim=0 0 0 0]{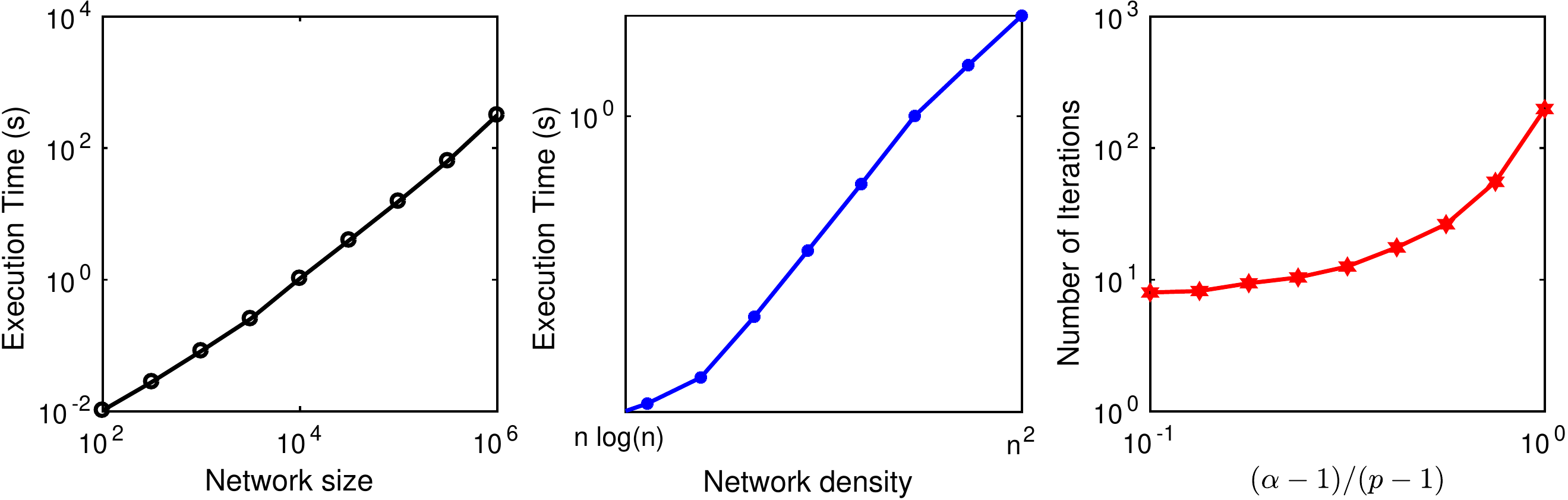}
\caption{\blue{All plots show mean values over 10 runs.  Left and center: Time required by Alg.\ 1 to converge to a tolerance of $10^{-8}$, with $\alpha=p/2=10$, for random Erd\H{o}s--Renyi graphs. Left: $n$ nodes and $m=O(n\log n)$ edges, with $n\in[10^2,10^6]$; Center: $n=1000$ nodes and $m \in [n\log n, n^2]$ edges. Right: Number of iterations required by Alg.\ 1 when the ratio $(\alpha-1)/(p-1)$ varies.} }\label{fig:NSM_timing}
\end{figure}

 Further comments on the algorithm above are in order. 
 First, recall that the convergence 
is independent of the starting point, $\b x_0$, provided that $\b x_0$ is entry-wise positive. 
In practice, we use a uniform vector. Concerning the choice of $\alpha$, recall that we want $\alpha$ large enough to give a good approximation to the original kernel $\max\{|x|,|y|\}$. 
As quantified in \eqref{eq:alpha_approx}, the approximation error is bounded by a factor 
$2^{1/\alpha}$.
Thus, in practice, moderate values of the parameter are sufficient. In order to avoid numerical issues, in the experiments presented in Section \ref{sec:experiments} we use $\alpha = 10$. 
As for the choice of $p$, from the proof of Theorem \ref{thm:convergence} it 
follows that the larger $p>\alpha$, the smaller the contraction ratio $C= (\alpha-1)/(p-1)$ and thus the faster the convergence of Algorithm~\ref{alg:1}. This is made more precise in
 Corollary~\ref{cor:conv_bound} below, where we explicitly bound  
$\|\b x_k -\b x_{k+1}\|$ and $\|\b x_k - \b x^\star\|$ in terms of $C$.
  Finally, the choice of the norm in the stopping criterion is not critical. 
We typically use the $p$-norm because the sequence $\b x_k$ is designed so that $\|\b x_k\|_p=1$ for any $k$. Hence in the stopping criterion we require one norm computation fewer at each step. However, this reduction in cost is likely to be negligible and thus we expect that
other distance functions would work equally well. Moreover, we point out that, due 
to Corollary~\ref{cor:conv_bound}, 
a computationally cheaper stopping criterion is available from the contraction ratio $(\alpha-1)/(p-1)$ and its integer powers. However, in our experience, this upper bound on the iteration error can be far from sharp.
  \begin{corollary}\label{cor:conv_bound}
For $\b x_0>0$, let $\b x_k$ be the sequence defined by Algorithm~\ref{alg:1} and let $\gamma = \|\log(\b x_1)-\log(\b x_0)\|_\infty$. For any $k=0,1,2,\dots$ we have 
$$
\|\b x_{k+1} -\b x_k\|_\infty\leq \gamma\, \left(\frac{\alpha-1}{p-1}\right)^k \quad \text{and}\quad \|\b x_{k} -\b x^\star\|_\infty\leq \gamma\, \left(\frac{p-1}{p-\alpha}\right)\left(\frac{\alpha-1}{p-1}\right)^k,	
$$
where $\b x^\star =\lim_k\b x_k$ is the unique positive solution of \eqref{eq:maximization_problem}.  
\end{corollary}
\begin{proof}
From the Mean Value Theorem we have $|e^a -e^b|\leq |a-b|\max\{e^a,e^b\}$. Thus, for any $\b x,\b y >0$ with $\|\b x\|_p=\|\b y\|_p=1$, we have 
$$\|\log(\b x) -\log(\b y)\|_\infty \geq \|\b x - \b y\|_\infty \Big(\max_i (\max\{x_i,y_i\})\Big)^{-1} \geq \|\b x - \b y\|_\infty, 
$$ as both $x_i$ and $y_i$ are not larger than one, for any $i=1,\dots, n$.  With the notation of the proof of Theorem \ref{thm:convergence}, this implies that $d_T(\b x, \b y)\geq \|\b x - \b y\|_\infty$ for any $\b x, \b y\in \mathcal S_p^+$. Moreover, from the proof of that theorem we have that $d_T(G_\alpha(\b x), G_\alpha(\b y))\leq C d_T(\b x, \b y)$, with $C=(\alpha-1)/(p-1)$. Therefore, as $\b x_k\in \mathcal S_p^+$ for any $k$, we have 
\begin{align*}
\|\b x_{k+1} -\b x_k\|_\infty&\leq d_T(\b x_{k+1},\b x_k) = d_T(G_\alpha(\b x_{k}),G_\alpha(\b x_{k-1}))\\
&\leq C d_T(\b x_k, \b x_{k-1}) \leq C^k d_T(\b x_1, \b x_0)\, .
\end{align*}
This proves the first inequality. As for the  second one, first note that it is enough to show that $d_T(\b x_k, \b x^\star)\leq (1-C)^{-1}d_T(\b x_{k+1},\b x_k)$ as we then can argue as before to obtain $\|\b x^\star-\b x_k\|_\infty\leq (1-C^{-1})C^kd_T(\b x_1,\b x_0)$ which is the right-most inequality in the statement. Now, observe that adding $d_T(\b x_{i+1}, \b x_{i})$ to both sides of the inequality $d_T(\b x_{i+2}, \b x_{i+1})\leq C d_T(\b x_{i+1}, \b x_{i})$ and rearranging terms leads to  
$$
d_T(\b x_{i+1}, \b x_{i})\leq \frac 1 {1-C}\Big(d_T(\b x_{i+1}, \b x_{i})- d_T(\b x_{i+2}, \b x_{i+1})\Big)\, ,
$$ 
for any $i=0,1,2,\dots$. 
Therefore, using the triangle inequality for $d_T$, for any $k, h$ with $h>k$, we obtain 
$$
d_T(\b x_{h+1}, \b x_{k})\leq \sum_{i=k}^h d_T(\b x_{i+1}, \b x_{i})\leq \frac{1}{1-C}\Big(d_T(\b x_{k+1}, \b x_{k})- d_T(\b x_{h+2}, \b x_{h+1})\Big)\, .
$$ 
Finally, letting $h$ grow to infinity in the previous inequality gives the desired bound and concludes the proof.
\end{proof}

\section{Experiments}\label{sec:experiments}
In this section we describe results obtained when the logistic core--periphery score computed via Algorithm \ref{alg:1} is used to rank nodes in some example networks. 
All experiments were performed using MATLAB Version
9.1.0.441655 (R2016b) on a laptop running Ubuntu 16.04 LTS with a  3.2 GHz Intel Core i5 processor and 8 GB of RAM.
The experiments can be reproduced using the code available at \verb+https://github.com/ftudisco/nonlinear-core-periphery+.

\begin{rev}
We compare results with those obtained from other core-quality function optimization approaches: the degree vector, the Perron eigenvector of the adjacency matrix (eigenvector centrality) and the simulated--annealing algorithm proposed in \cite{rombach2014core}. Furthermore, we compare with the $k$-core decomposition coreness score \cite{kitsak2010identification}, computed as the limit of the $H$-index operator sequence discussed in \cite{lu2016h}. 

The use of the degree vector $\b d$ and the eigenvector centrality $\b v$ may be regarded as  linear counterparts of our method.
If $\alpha = 1$, then for any $\b x \geq 0$ the functional $f_\alpha(\b x) = \sum_{ij}a_{ij}\mu_\alpha(x_i,x_j)$ is linear and has the form  $f_1(\b x) = \sum_{ij}a_{ij}(x_i+x_j) = 2\,\b x^T \b d$. Thus the maximum is attained when $\b x$ is the degree vector $\b d$. The eigenvector centrality $A\b v = \rho(A)\b v$, $\|\b v\|_2=1$,  instead, somewhat corresponds to the case where $\alpha$ goes to $0$. To obtain $\b v$, however, we need to slightly modify the approximate kernel $\mu_\alpha$ from \eqref{eq:falpha} to 
$$
\tilde \mu_\alpha(x,y) = \left(\frac{|x|^\alpha + |y|^\alpha}{2}\right)^{1/\alpha} \, .
$$
This is because $\mu_\alpha$ diverges when $\alpha\to 0$, whereas $\lim_{\alpha\to 0}\tilde \mu_\alpha(x,y) = \sqrt{|xy|}$. On the other hand, notice that both $\mu_\alpha$ and $\tilde \mu_\alpha$ coincide with the maximum operator when $\alpha\to \infty$ and for any fixed $\alpha>0$, a vector that maximizes $\sum_{ij}A_{ij}\tilde \mu_\alpha(x_i,x_j)$ maximizes $f_\alpha$ as well. Replacing $\mu_\alpha$ with $\tilde \mu_\alpha$ we have $f_0(\b x) = \sqrt{\b x}^T A \sqrt{\b x}$. Thus, if we choose $p=1$, the maximum is attained when $\b x = \b v^2$, the square of the entry-wise positive eigenvector centrality. Note that with this choice of $p$, the solution $\b v$ is constrained on the Euclidean sphere $\|\b v\|_2=1$. Notice moreover that this is confirmed by Theorem \ref{thm:mlp} as,  for $\alpha\to 0$ and $p=1$, the nonlinear operator $F_\alpha$ boils down to the matrix $A$ and Algorithm 1 is equivalent to the standard linear power method. 
\end{rev}

As for the simulated--annealing method, recall that it aims to 
maximize the core-quality function $\sum_{ij}a_{ij}x_ix_j$ over $C_{\alpha,\beta}$, defined as in \eqref{eq:C_alpha_beta}. To this end the method requires a uniform discretization of the square $[0,1]^2$. In all our experiments below we choose the discretization $\{1/h,2/h,\dots,1\}^2$ with $h=50$. 

Algorithm \ref{alg:1} 
requires the selection of two positive 
scalars, $\alpha$ and $p>\alpha$, and the norm in the stopping criterion. In all our experiments 
we set $\alpha = 10$, $p = 2 \alpha$, 
and terminate when
$$
\frac{\|\b x_k-\b x_{k+1}\|_p }{\|\b x_{k+1}\|_p}=\|\b x_k-\b x_{k+1}\|_p < 10^{-8}.
$$

\begin{rev}
For the sake of brevity, we refer to the nonlinear spectral method,  simulated--annealing method,  
degree--based method, eigenvector centrality method, and the $H$-index $k$-core decomposition method as NSM, Sim-Ann, Degree, Eig, and Coreness respectively.  
We point out that, in order to reduce the computing time, we implement Sim-Ann in parallel on four cores, whereas all other methods are run on a single computing core. 
\end{rev}

\subsection{Synthetic Networks} 
In practice, of course, it is typically not known ahead of time whether a given network
contains any inherent core--periphery structure. 
However, in order to conduct a set of controlled tests, we begin with 
two classes of random networks that have a built-in core--periphery structure. 
The first takes the form of a stochastic block model, a widespread benchmark where community structure is imposed in block form. 
We then consider the new logistic core--periphery random model discussed in 
section~\ref{sec:logistic_random_model}.
\begin{rev}
For the sake of brevity, we only compare NSM, Sim-Ann and Degree
in these synthetic tests, noting that  
Eig and Coreness were comparable to or less effective than Sim-Ann.
\end{rev}

\begin{figure}[t]
\centering
\includegraphics[width=.3\textwidth]{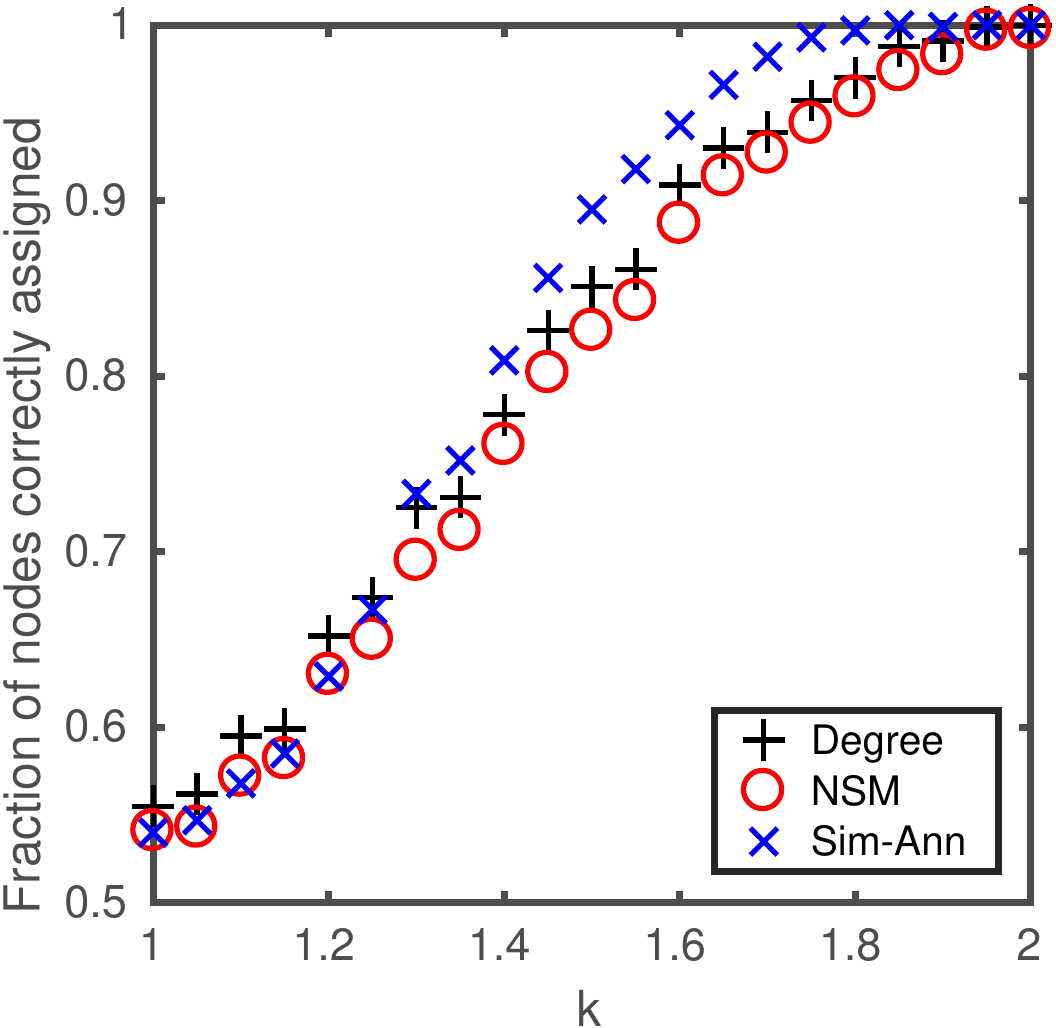}
\includegraphics[width=.3\textwidth]{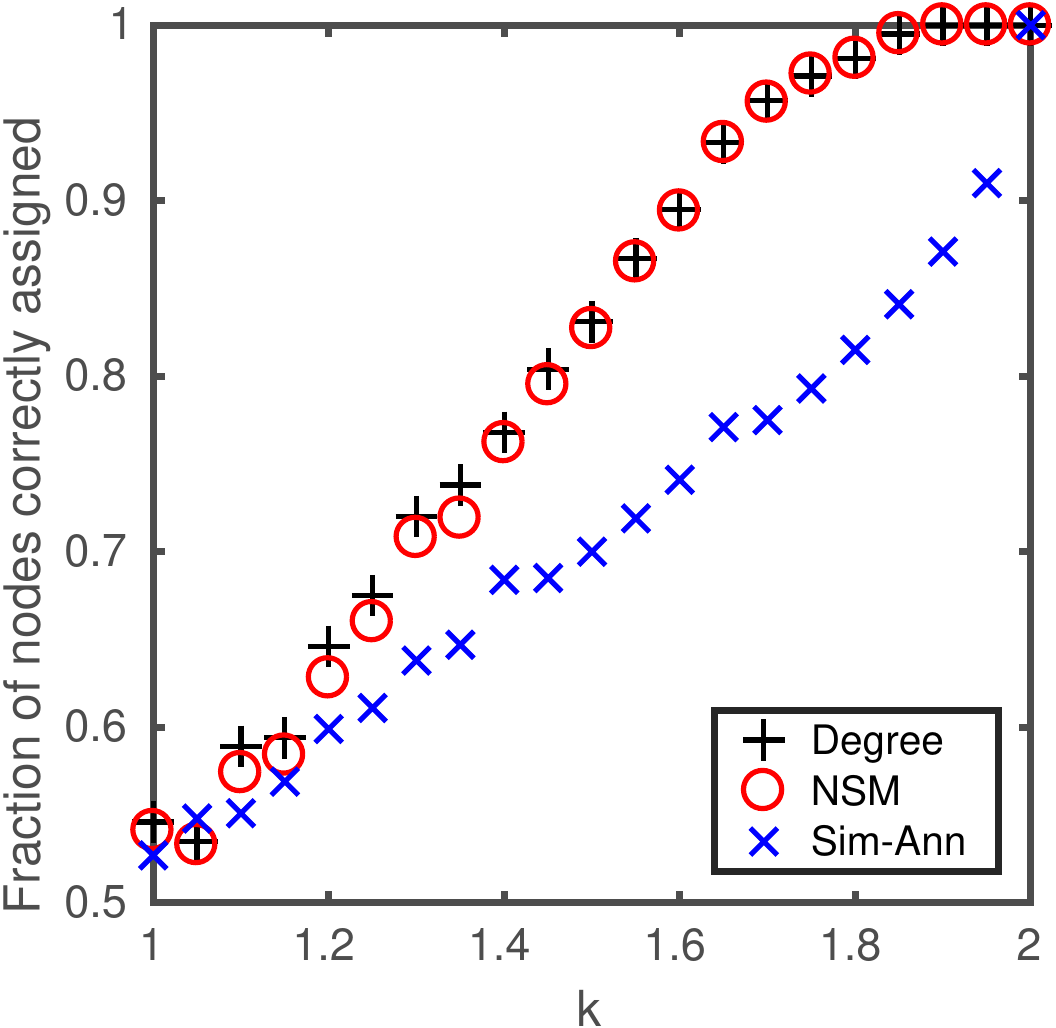}
\includegraphics[width=.315\textwidth]{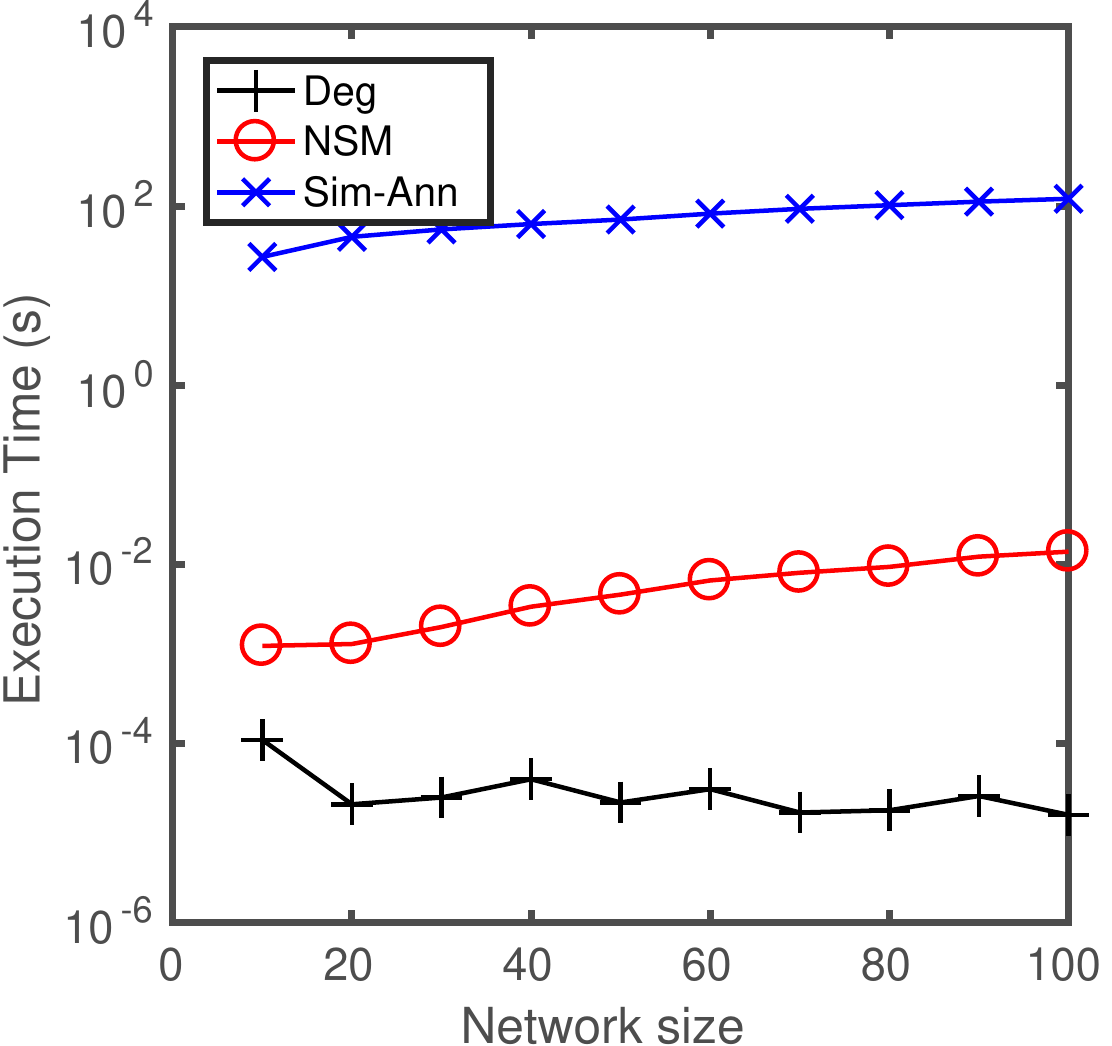}
	\caption{(Color online.) Experiments on stochastic block model graphs. Left and center: Fraction of nodes correctly assigned to the core--periphery ground truth versus model parameter $k$,  for three methods: our Nonlinear Spectral Method (red circles), the Degree vector (black plus symbols) and the Simulated Annealing technique of Rombach et al. \cite{rombach2017core} (blue crosses). Each value is the mean over 20 random instances. Each network has 100 nodes and $k$ ranges in $\{1,1.05,1.1,\dots,2\}$. Left: Nodes within the periphery and between core and periphery are connected with probability $k/4$, nodes within the core are connected with probability $k^2/4$. Center: Nodes within the core and between core and periphery are connected with probability $k^2/4$, nodes within the periphery are connected with probability $k/4$. Right: Median execution time of the three methods over 20 instances, when the number of nodes varies within $\{10,20,\dots,100\}$.}\label{fig:SBM}
\end{figure}

\subsubsection{Stochastic Block Model}\label{sec:SBM}
We consider synthetic networks that have a planted core--periphery structure, arising 
from a stochastic block model. For the sake of consistency with previous works, we denote this ensemble of unweighted networks by $\mathrm{CP}(n,\delta,p,k)$. Each network drawn has $\delta n$ core nodes and $(1-\delta)n$ periphery nodes, with $\delta \in [0,1]$. We consider two parameter settings. The first reproduces the case analyzed in \cite[Sec.\ 5.1]{rombach2017core}: for $p\in[0,1]$ and $k\in [1,1/\sqrt{p}]$, each edge between nodes $i$ and $j$ is assigned independently at random with probability $kp$, if either $i$ or $j$ (or both) are in the periphery and with probability $k^2p$ if both $i$ and $j$ are in the core. In the second setting, edges between nodes $i$ and $j$ have probability $kp$ only if both $i$ and $j$ are in the periphery and have probability $k^2p$ otherwise.

In our experiment we fix $n=100$, $\delta=1/2$, $p=1/4$ and, for each $k=\{1,1.05,1.1,\dots,2\}$, we compute the core--periphery assignment for a network drawn from $\mathrm{CP}(n,\delta,p,k)$.  Figure~\ref{fig:SBM}  shows the percentage of nodes correctly assigned to the ground-truth core--periphery structure in the two settings described above (left and central figures) by NSM (red circles),
Sim-Ann (blue crosses) and 
Degree (black plus symbols). Each plot 
shows the mean over 20 random instances of 
$\mathrm{CP}(n,\delta,p,k)$, for each fixed value of $k$. The right plot in the figure, instead,  shows the median execution time of the three methods over 20 runs. 
We see that all three approaches give similar results in the first parameter regime, whereas
Degree and NSM outperform Sim-Ann in the second. Using the degree gives the cheapest method, and Sim-Ann is around two orders of magnitude more expensive than NSM.

\begin{figure}[t]
\centering
\includegraphics[width=.9\textwidth]{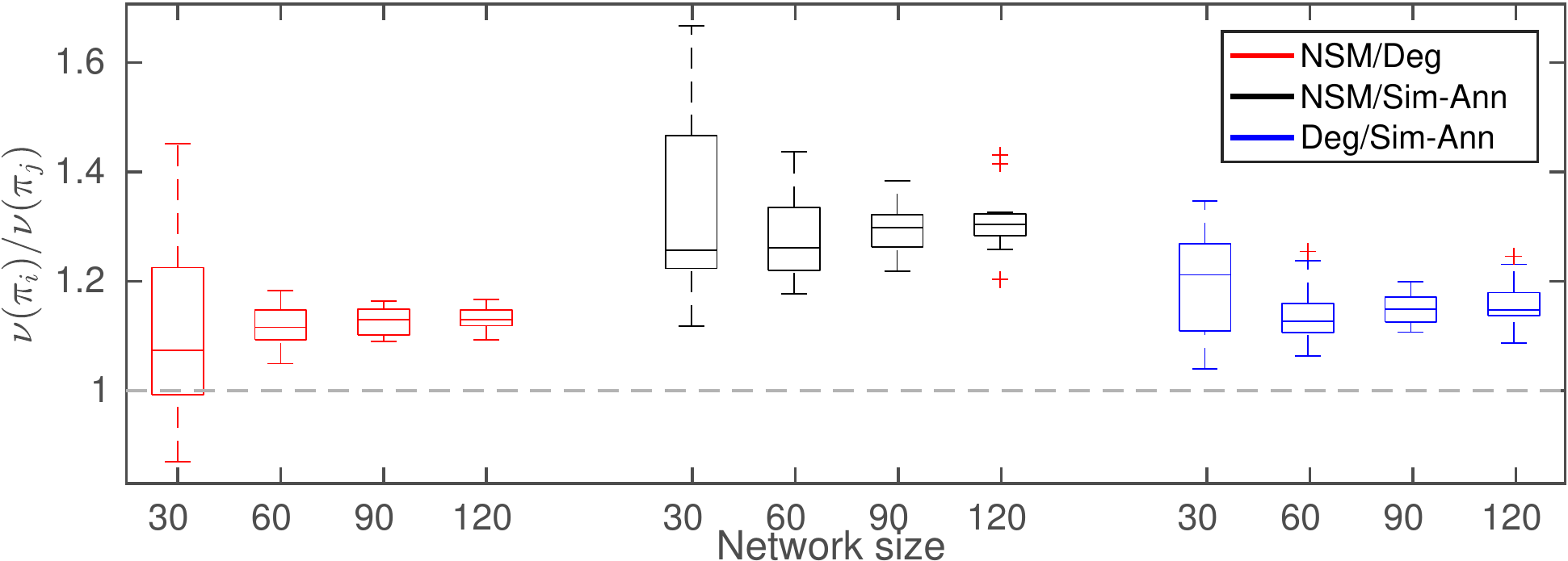}
\caption{(Color online.) Boxplots of the ratio of the likelihood $\nu(\b \pi)$ over 20 trials for different sizes $n$ of the random network, ranging within $\{30,60,90,120\}$. Three permutation vectors $\b \pi_1$, $\b \pi_2$, $\b \pi_3$ are obtained by sorting the entries of the score vectors obtained with: (1) the proposed Algorithm \ref{alg:1} (``NSM'' in the legend), (2) the degree vector of the graph (``Deg'' in the legend) and (3) the simulated--annealing method of \cite{rombach2014core} (``Sim-Ann'' in the legend), respectively. The three boxplot groups, with different colors, show (from left to right) the ratios $\nu(\b \pi_1)/\nu(\b \pi_2)$ (in red), $\nu(\b \pi_1)/\nu(\b \pi_3)$ (in black) and $\nu(\b \pi_3)/\nu(\b \pi_2)$ (in blue).}\label{fig:boxplot}
\end{figure}
\subsubsection{Logistic Core--periphery Random Model}
We now consider the unweighted logistic core--periphery random model described in 
subsection~\ref{sec:logistic_random_model}. More precisely, given $n$, $s$ and $t$, we
sample from the family $\mathrm{LCP}(n,s,t)$ of random graphs with $n$ nodes 
such that an edge between any pair of nodes $i$ and $j$ is assigned independently at random with probability 
$$
\textstyle{\mathbb P(i\sim j) = \sigma_{s,t}\left(\frac 1 n \max\{n-i,n-j\}\right), \quad \text{where}\quad \sigma_{s,t}(x)=\frac{1}{1+e^{-s(x-t)}}\, .}
$$

Unlike the stochastic block model discussed in Section \ref{sec:SBM}, if $s$ is not too large the logistic core--periphery model does not give rise to a binary core--periphery structure. Instead, 
it uses a sliding scale for the nodes where node $n$ is at the center of the core and node $1$ is the most peripheral. We therefore look at the ability of the algorithms to recover a suitable ordering. 

In our experiment we fix $s=7$, $t=2/3$ and let the dimension $n$ vary within $\{30,60,90,120\}$. For each $n$ we draw an instance from the ensemble $\mathrm{LCP}(n,s,t)$
and compute the core--periphery score from
each of the three methods.
We sort each score vector into descending order and consider the associated permutations $\b \pi_1$, $\b \pi_2$ and $\b \pi_3$ 
for NSM, Degree and Sim-Ann, respectively.
We then evaluate the likelihoods $\nu(\b \pi_i)$, as defined in \eqref{eq:LCP_probability}. 
In Figure~\ref{fig:boxplot} we show medians and quartiles of the three 
likelihood 
ratios $\nu(\b \pi_1)/\nu(\b \pi_2)$ (in red), $\nu(\b \pi_1)/\nu(\b \pi_3)$ (in black) and $\nu(\b \pi_3)/\nu(\b \pi_2)$ (in blue).
We see that in this test NSM outperforms Degree, which itself 
 outperforms Sim-Ann.

\subsection{Real-world Datasets}
\begin{rev}
In this subsection we show results 
for several real-world networks taken from different fields: social interaction,
academic collaboration, transportation, internet structure, neural connections, 
and protein-protein interaction. These networks are freely available online;
below we
describe their key features and give references for further details. 
\end{rev}

	\textbf{Cardiff tweets.} An unweighted network of reciprocated Twitter mentions among users whose bibliographical information indicates they are associated with the city of Cardiff (UK).
Data refers to the period October 1--28, 2014.  There is a single connected component of $2685$ nodes and $4444$ edges. The mean degree of the network is $3.31$, with a variance of $21.24$ and diameter $29$. This dataset is part of a larger collection of
geolocated reciprocated Twitter mentions within UK cities in \cite{grindrod2016comparison}.

	\textbf{Network scientists.} A weighted co--authorship network among scholars who study network science. This network, compiled in 2006, involves $1589$ authors.
We use the largest connected component, which has $379$ nodes and coincides with the network considered in \cite{newman2006finding}. This component contains $914$ edges. Its mean degree is $4.82$, with variance $15.46$ and diameter $17$.

	\textbf{Erd\H{o}s.} An instance of the Erd\H{o}s collaboration unweighted network with $472$ nodes representing authors. We use the largest connected component, which contains $429$ nodes and $1312$ edges. Its mean degree is $6.12$ with variance $45.98$ and diameter $11$. This dataset is one of the seven Erd\H{o}s  collaboration networks made available by Batagelj and Mrvar in the Pajek datasets collection \cite{pajek} and therein referred to as ``Erdos971''.  

	\textbf{Yeast.} An unweighted protein-protein interaction network described and analyzed in \cite{bu2003topological}. As for the Erd\H{o}s dataset, this network is available through the datasets collection \cite{pajek}. The whole network consists of $2361$ nodes. We use the 
largest connected component, consisting of $2224$ nodes and $6829$ edges. Its mean degree is $6.14$ with variance $65.76$ and diameter $11$. 

\begin{rev}
\textbf{Internet 2006.} A symmetrized snapshot of the structure of the Internet at the level of autonomous systems, reconstructed from BGP tables posted by the University of Oregon Route Views Project. This snapshot was created by Mark Newman from data for July 22, 2006 and is available via \cite{Internet} and \cite{davis2011university}. The network is connected, with $22963$ nodes and $48436$ edges. Its mean degree is $4.2186$ with variance $108.5$ and diameter $11$.

\textbf{Jazz.} Network of Jazz bands that performed between 1912 and 1940 obtained from ``The Red Hot Jazz Archive'' digital database \cite{gleiser2003community}. It consists of $198$ nodes, being jazz bands, and $2742$ edges representing common musicians. Mean degree is $27.7$, with variance $304.6$ and diameter $12$. 

\textbf{Drugs.} Social network of injecting drug users (IDUs) that have shared a needle
in a six months time-window. This is a connected network made of $616$ nodes and  $2012$ edges. The average degree is $6.5$ with variance $59.17$ and diameter $13$. See e.g.\ \cite{neaigus1998network,tudisco2017community}.

\textbf{C. elegans.} This is a neural network of neurons and synapses in Caenorhabditis elegans, a type of worm. It contains $277$ nodes and $2105$ edges.
Mean degree is $7.6$ with variance $48$ and diameter $6$. The network was created in \cite{Choe-2004-connectivity}. The data we used is collected from \cite{AustinData}.
\end{rev}

	\textbf{London trains.} A transportation network representing connections between train stations of the city of London. The undirected weighted network that we 
consider here is the aggregated version of the original multi-layer network. It consists of a single connected component with $369$ nodes, each corresponding to a train station. 
Direct connections between stations form a set of $430$ edges with nonzero weights.
Each such weight takes an integer value of 
$1$, $2$, or $3$ according to the number of different types of connection,  from the three possibilities of 
underground, overground and Docklands Light Railway (DLR). 
The average degree is $2.33$ with variance $1.04$ and diameter $39$. This network is studied in \cite{DDSBA14} and the data we used was collected from \cite{DataDD}.

\begin{figure}
\centering
\begin{tikzpicture}
\node at (0, 2) {\sf Degree};
\node at (4, 2) {\sf Sim-Ann};
\node at (8, 2) {\sf NSM};
\node[rotate=90] at (-2.5,0) {\sf Cardiff Tweets};
\node at (0,0) {\includegraphics[width=.25\textwidth,trim=0 0 0 .5cm,clip]{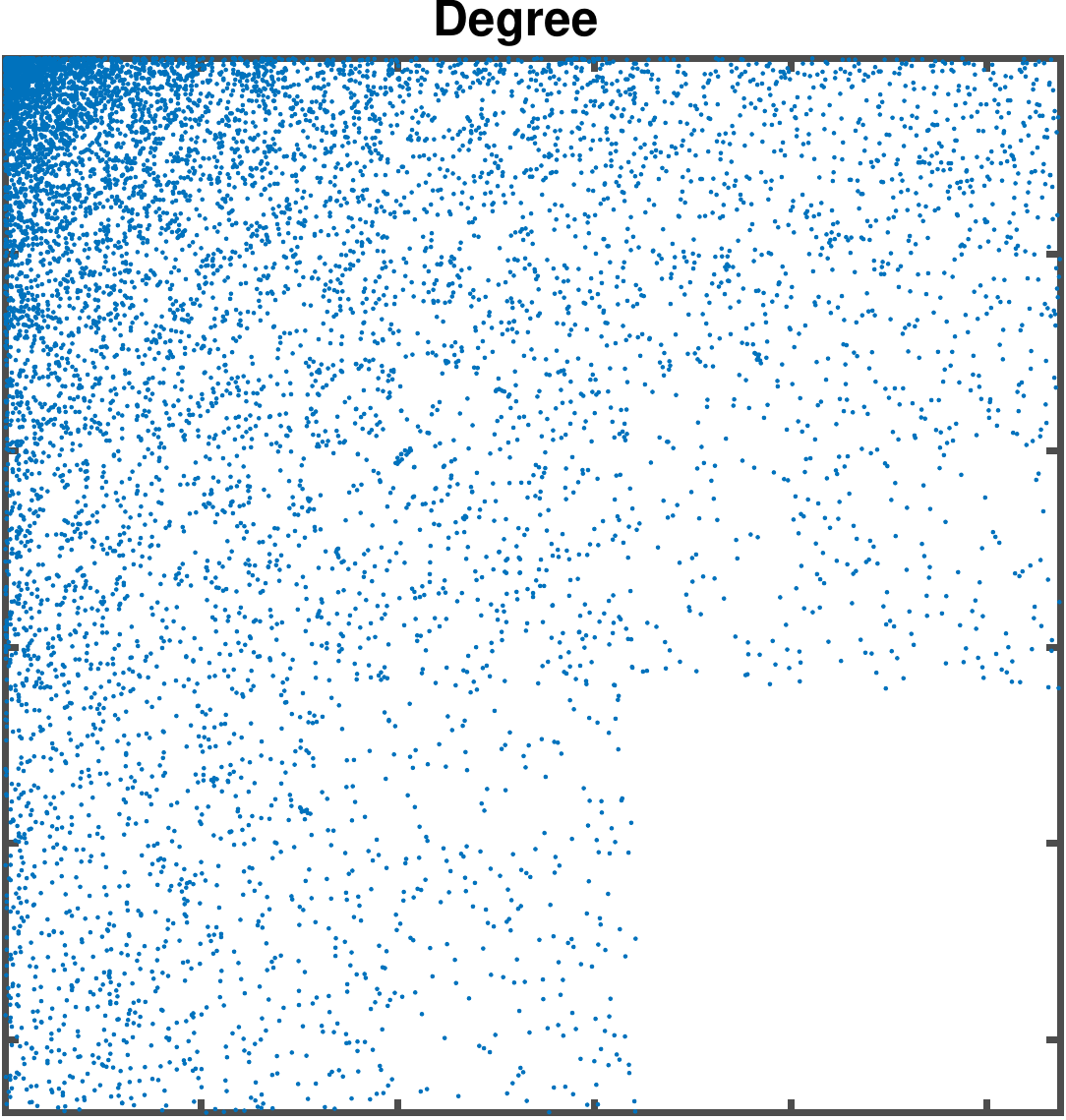}};
\node at (4,0) {\includegraphics[width=.25\textwidth,trim=1.3cm 1.2cm 0 .5cm,clip]{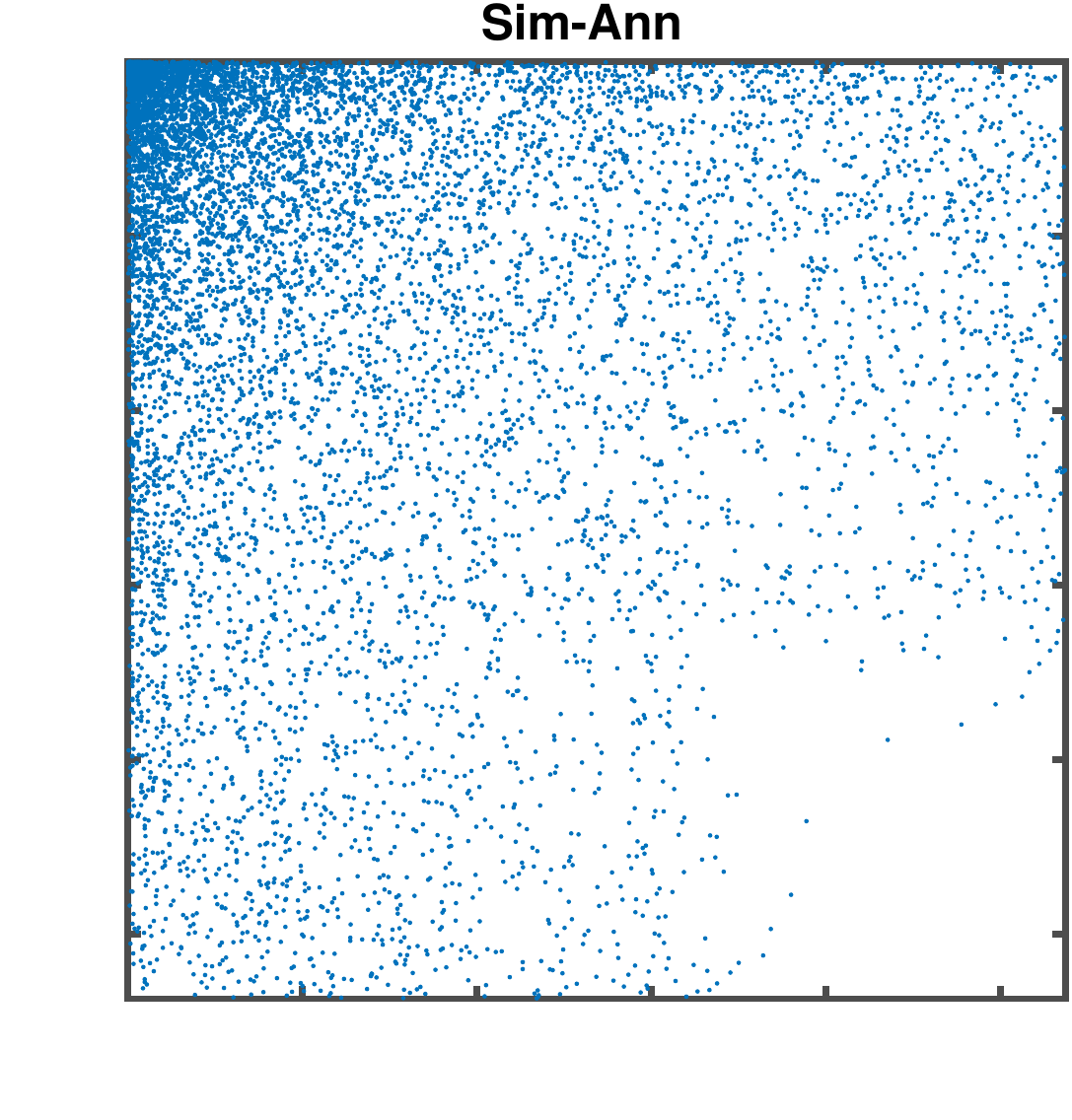}};
\node at (8,0) {\includegraphics[width=.25\textwidth,trim=0 0 0 .5cm,clip]{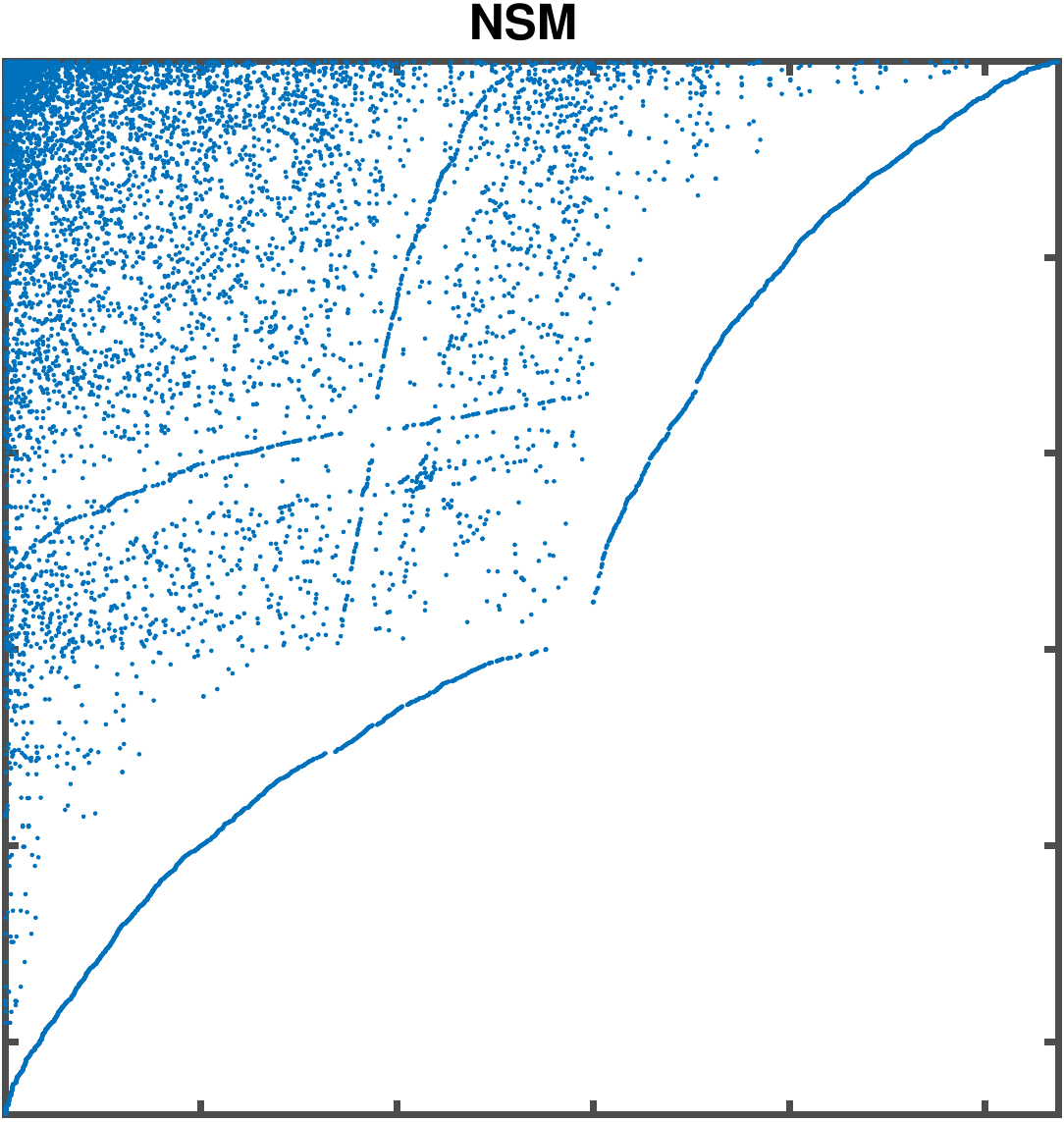}};
\node[rotate=90] at (-2.5,-3.5) {\sf Net.\ Scientists};
\node at (0,-3.5) {\includegraphics[width=.25\textwidth,trim=0 0 0 .5cm,clip]{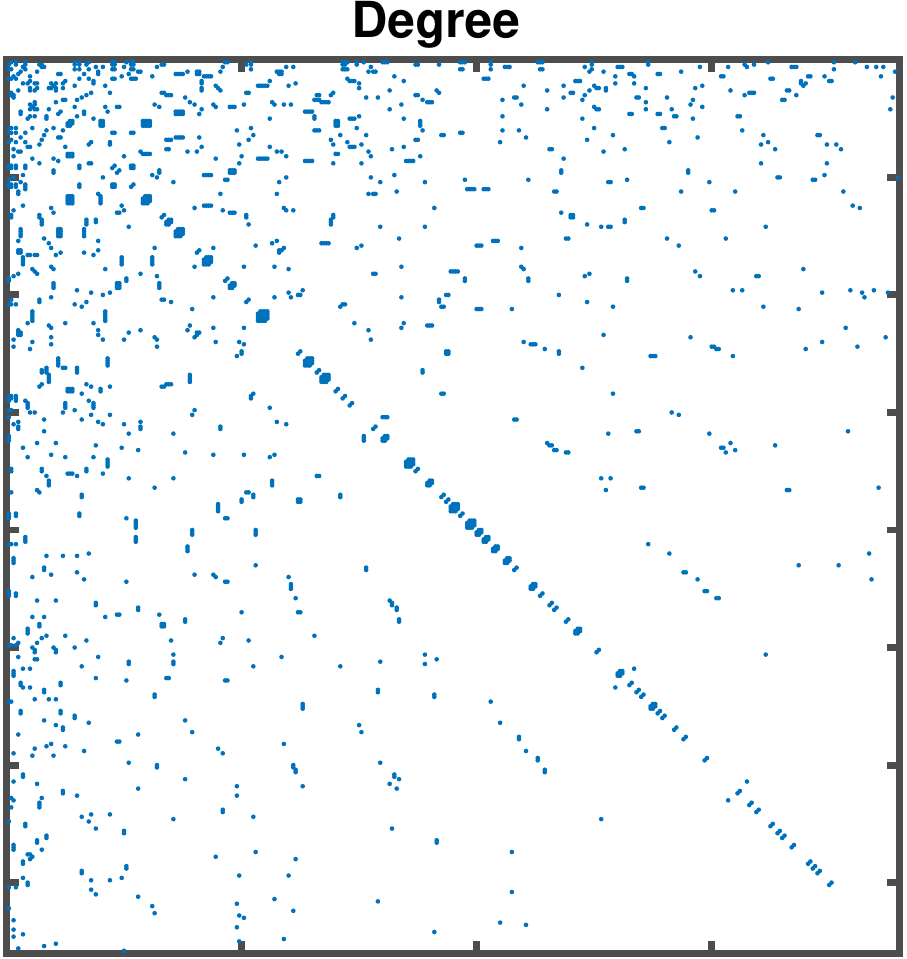}};
\node at (4,-3.5) {\includegraphics[width=.25\textwidth,trim=0 0 0 .5cm,clip]{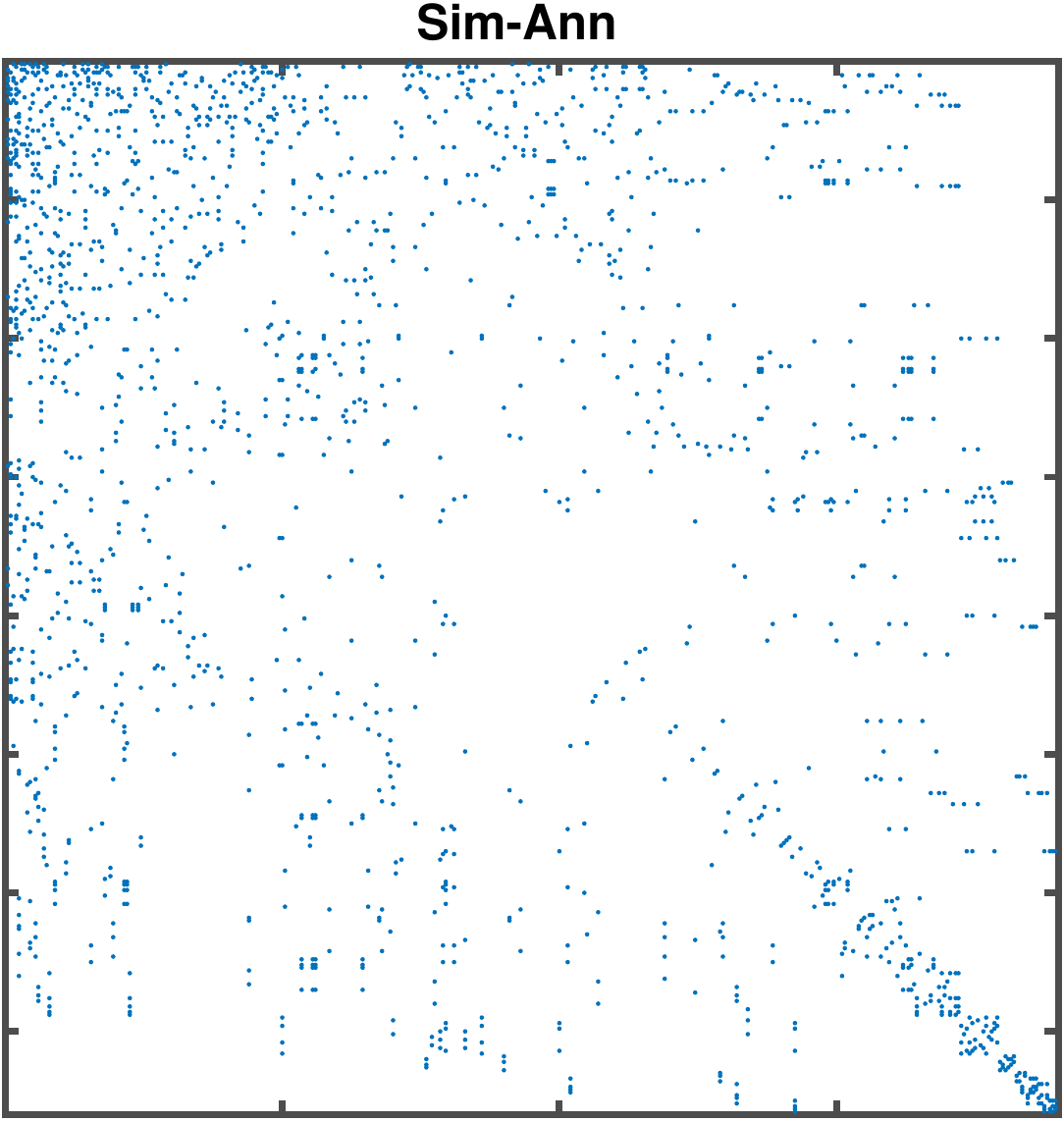}};
\node at (8,-3.5) {\includegraphics[width=.25\textwidth,trim=0 0 0 .5cm,clip]{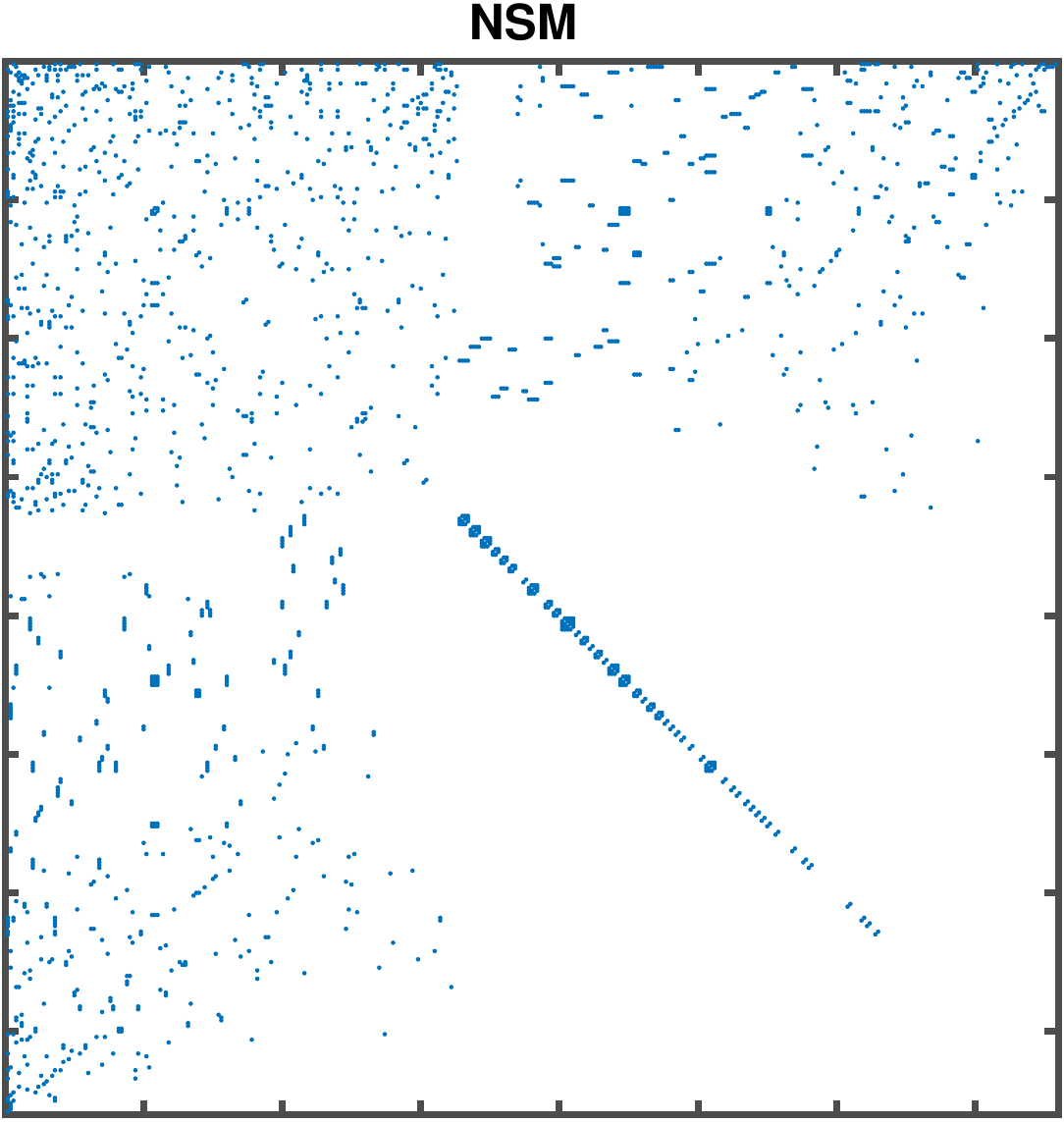}};
\node[rotate=90] at (-2.5,-7) {\sf Erd\H{o}s};
\node at (0,-7) {\includegraphics[width=.25\textwidth,trim=0 0 0 .5cm,clip]{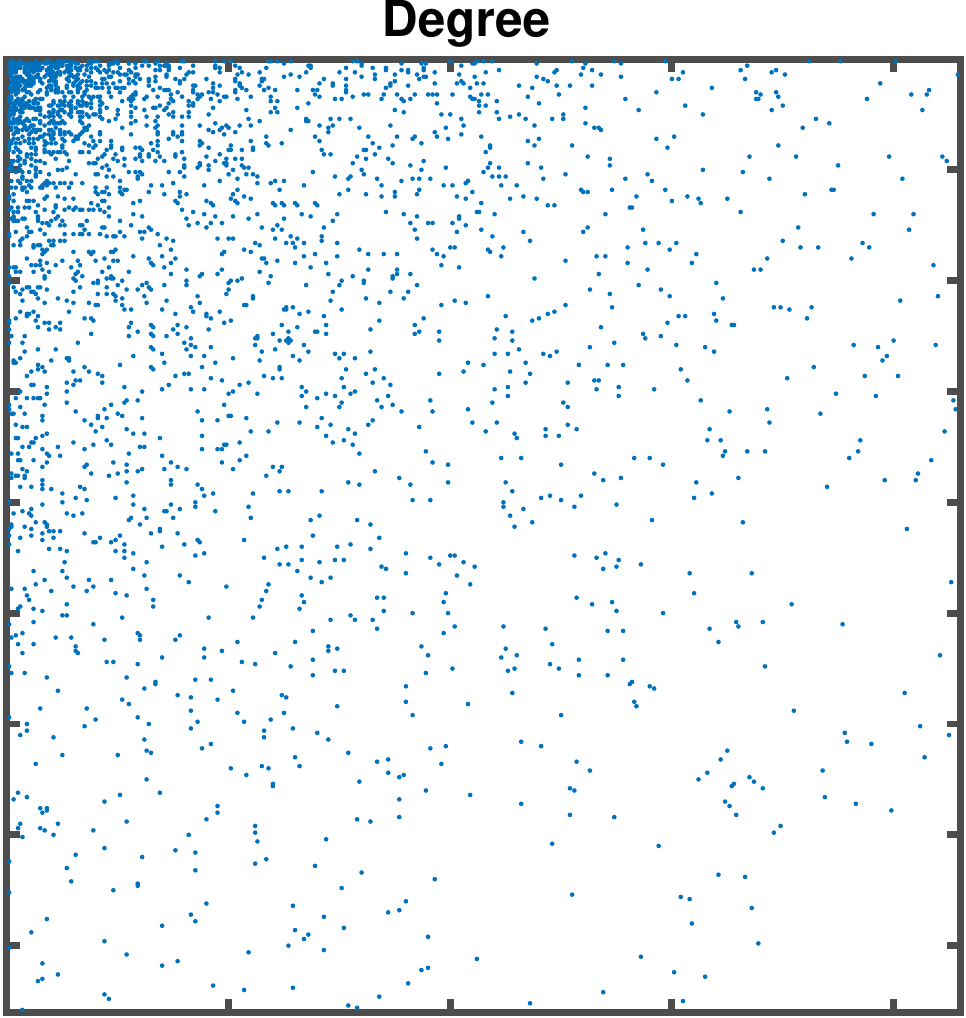}};
\node at (4,-7) {\includegraphics[width=.25\textwidth,trim=0 0 0 .5cm,clip]{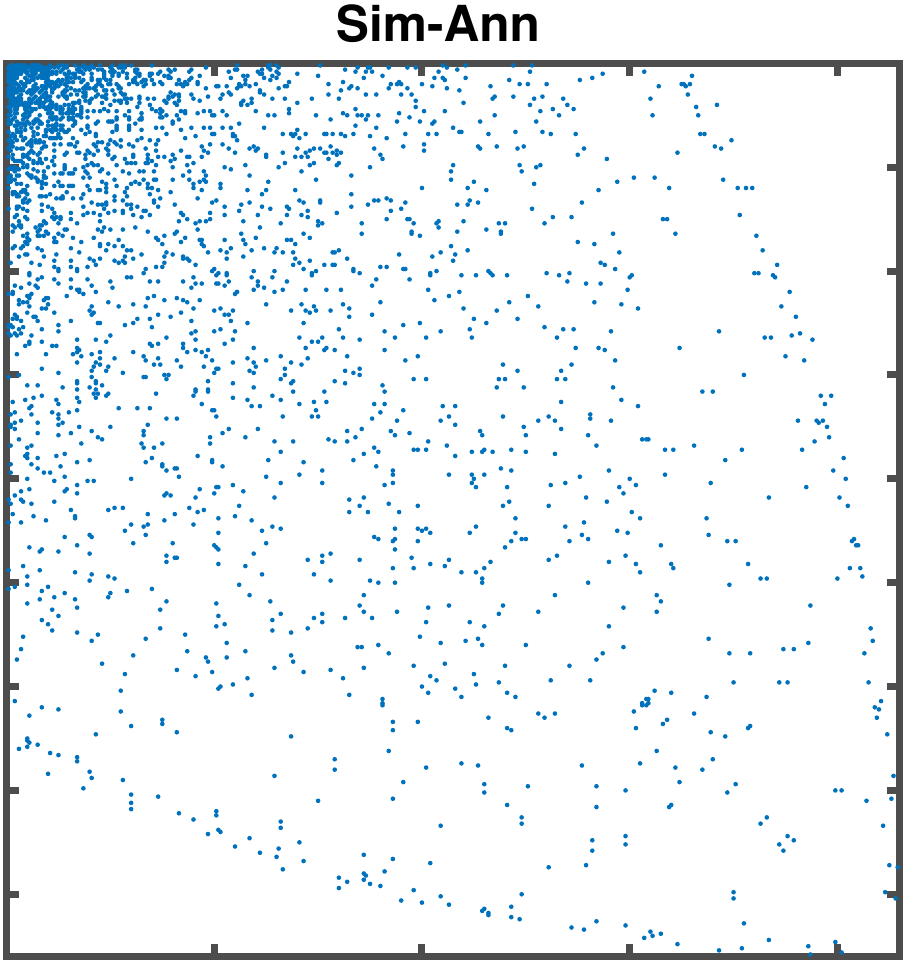}};
\node at (8,-7) {\includegraphics[width=.25\textwidth,trim=0 0 0 .5cm,clip]{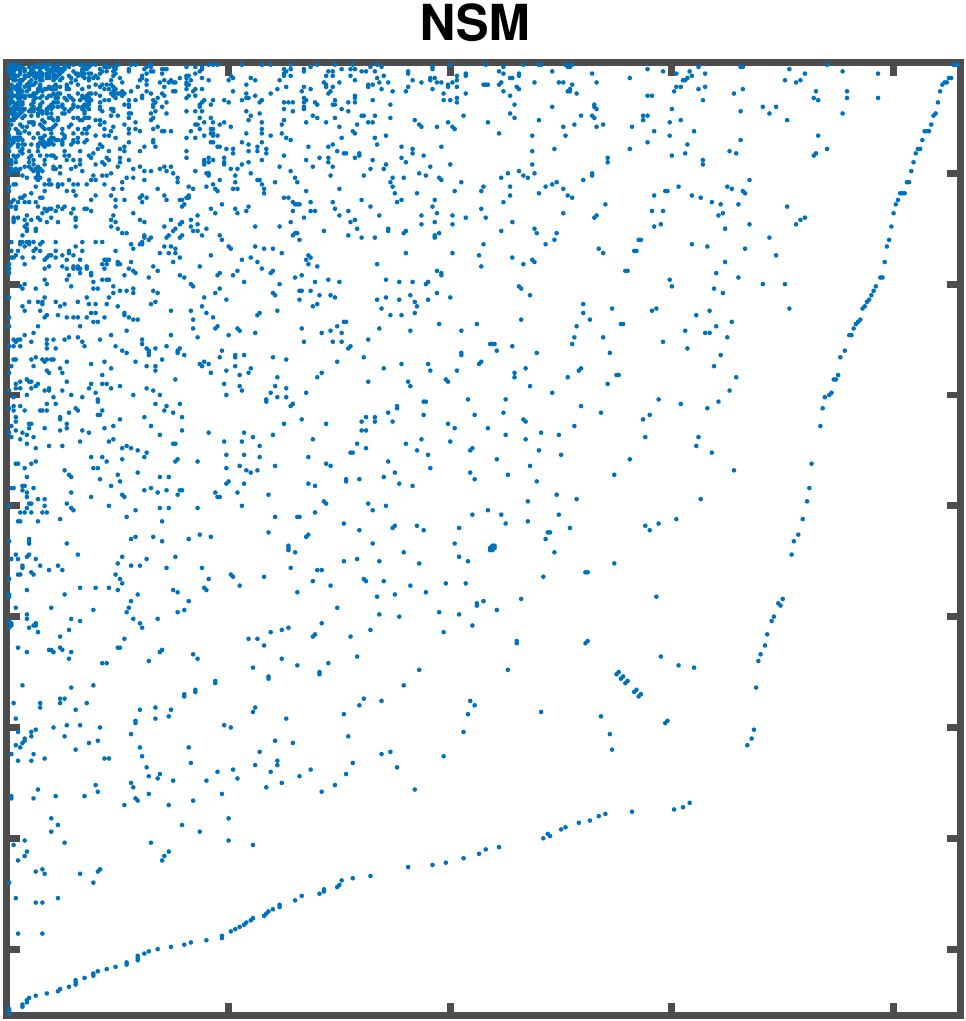}};
\node[rotate=90] at (-2.5,-10.5) {\sf Yeast PPI};
\node at (0,-10.5) {\includegraphics[width=.25\textwidth,trim=0 0 0 .5cm,clip]{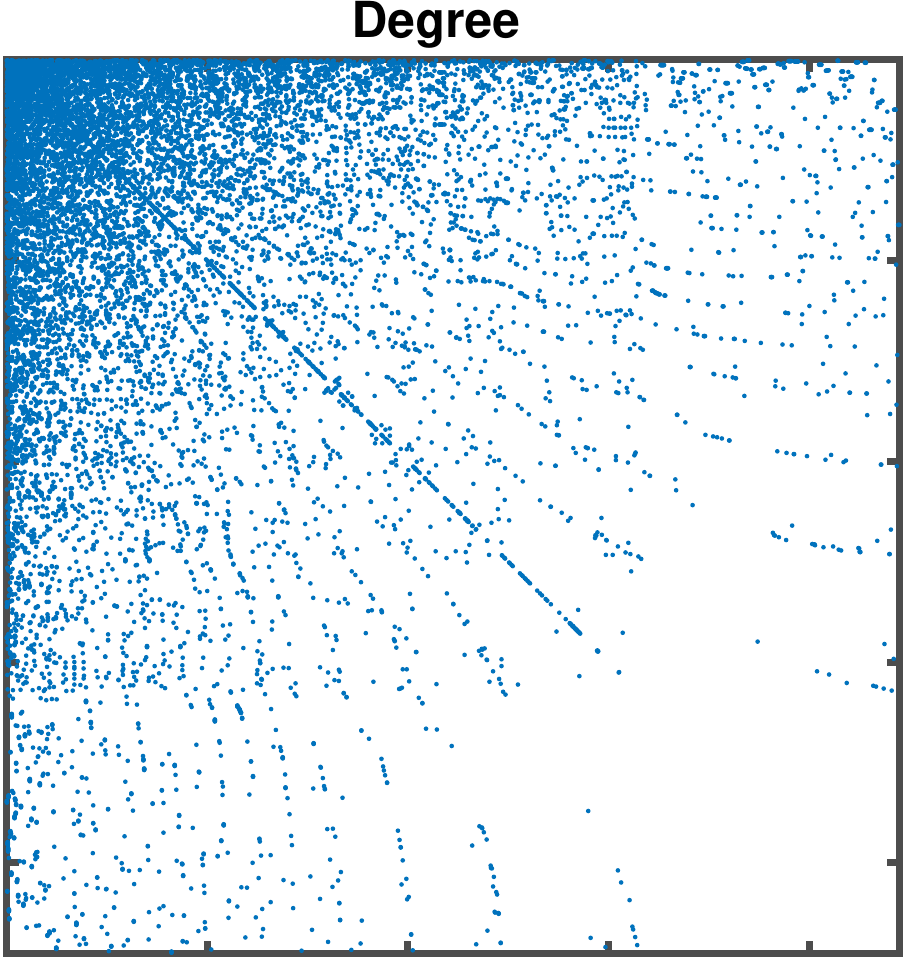}};
\node at (4,-10.5) {\includegraphics[width=.25\textwidth,trim=0 0 0 .5cm,clip]{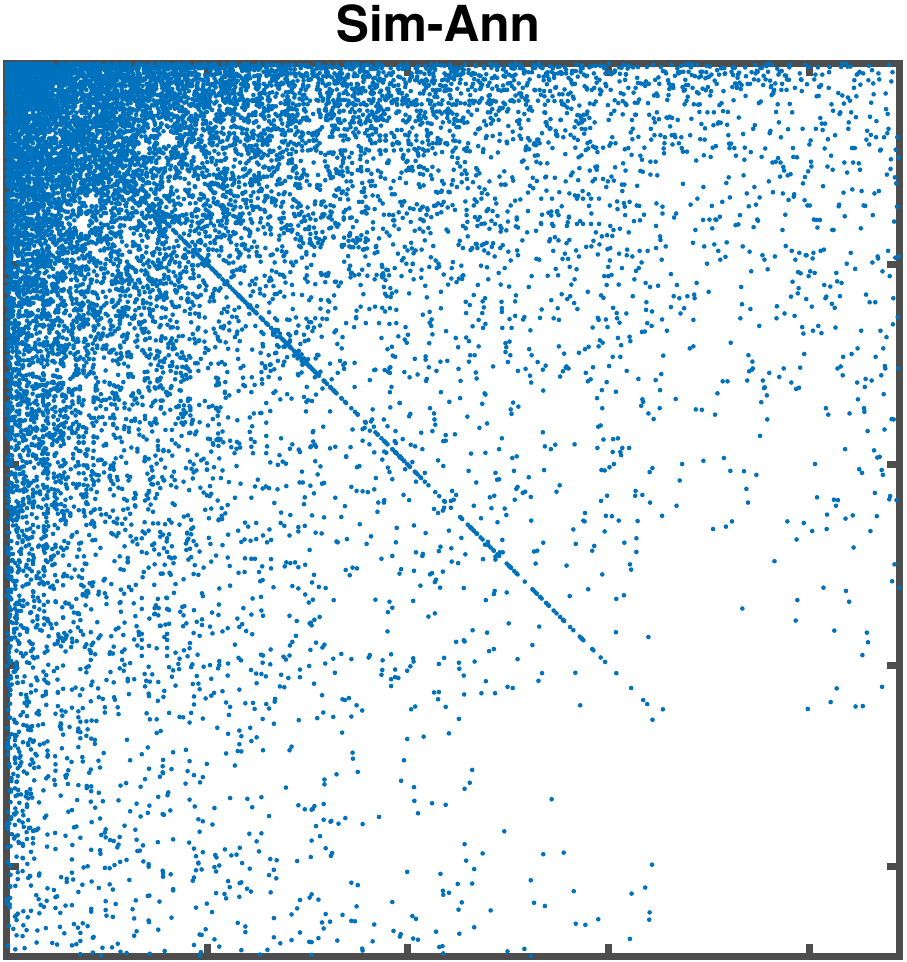}};
\node at (8,-10.5) {\includegraphics[width=.25\textwidth,trim=0 0 0 .5cm,clip]{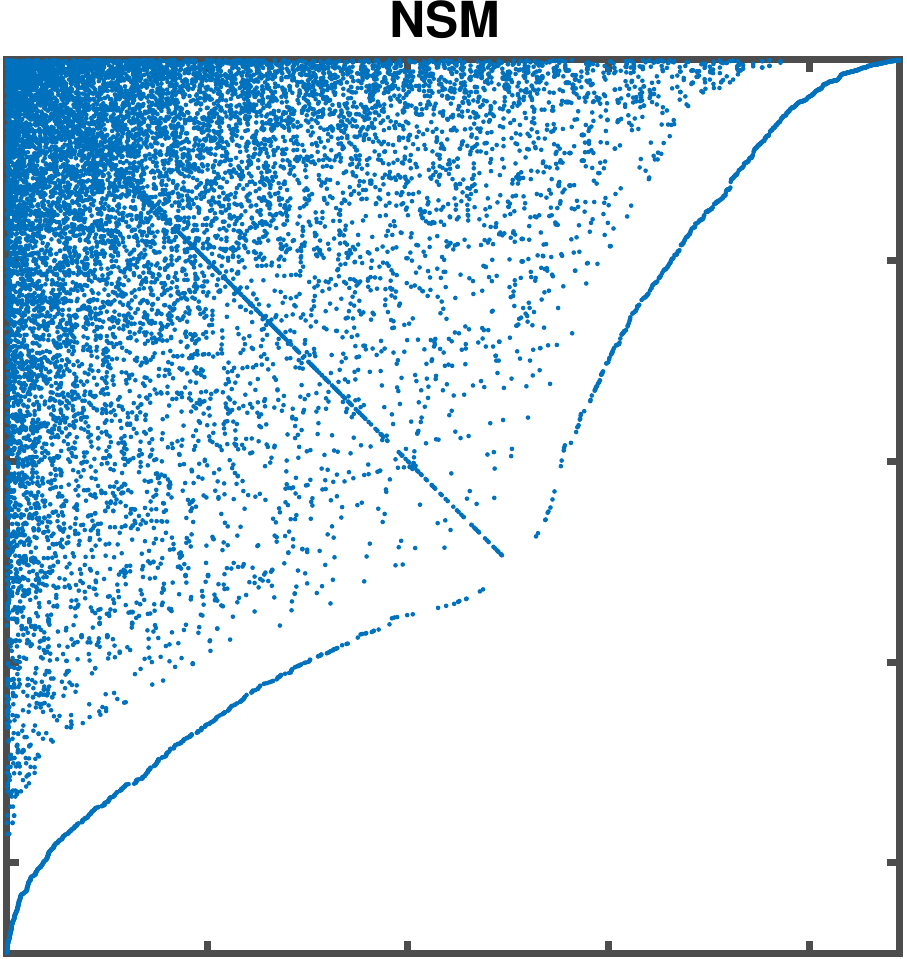}};
\node[rotate=90] at (-2.5,-14.3) {\sf Internet 2006};
\node at (0,-14.3) {\includegraphics[width=.25\textwidth,trim=3.05cm 0 2.6cm .6cm,clip]{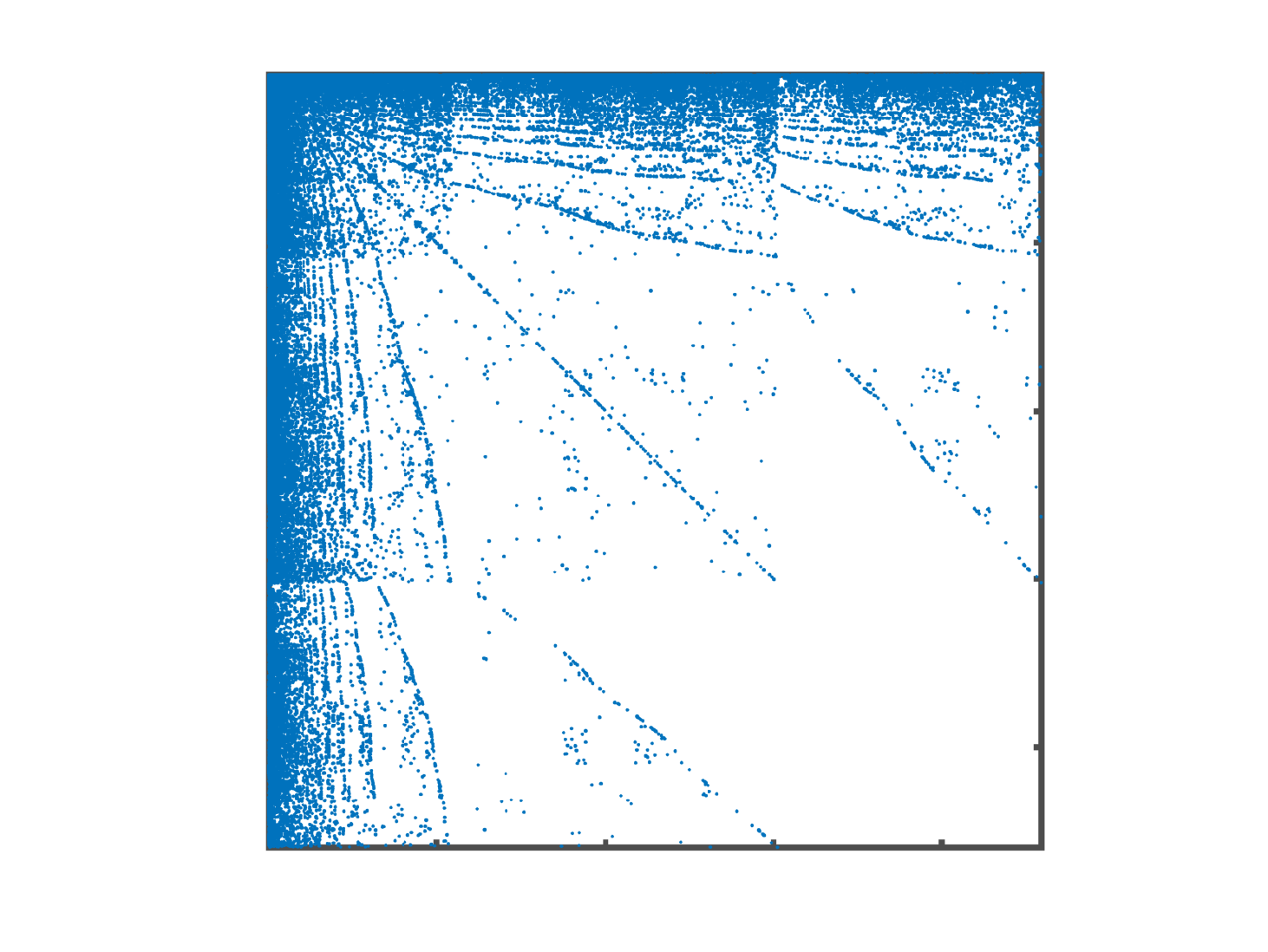}};
\node at (4,-14.3) {\includegraphics[width=.25\textwidth,trim=3.5cm 0 3cm .6cm,clip]{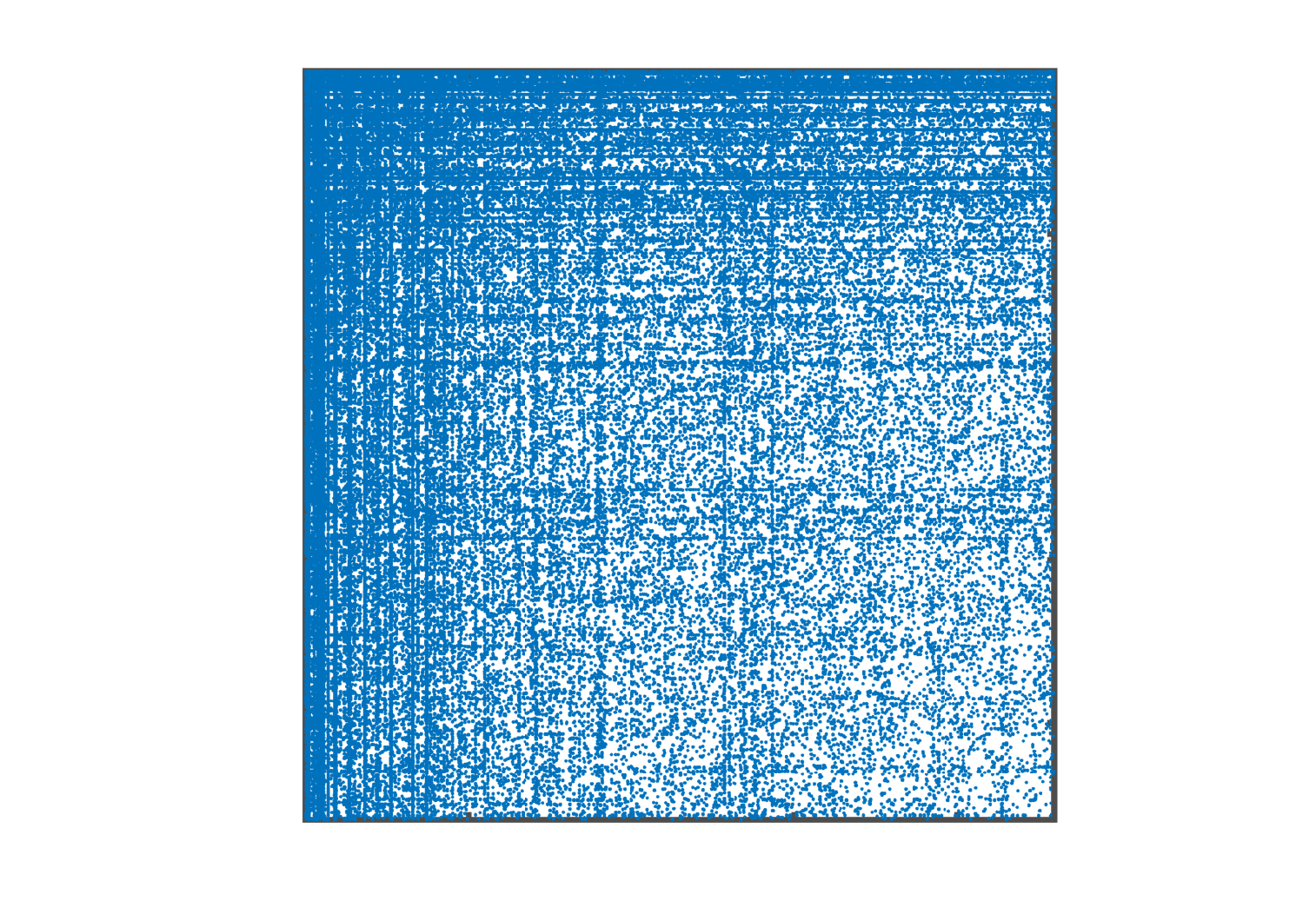}};
\node at (8,-14.3) {\includegraphics[width=.25\textwidth,trim=3.5cm 0 3cm .6cm,clip]{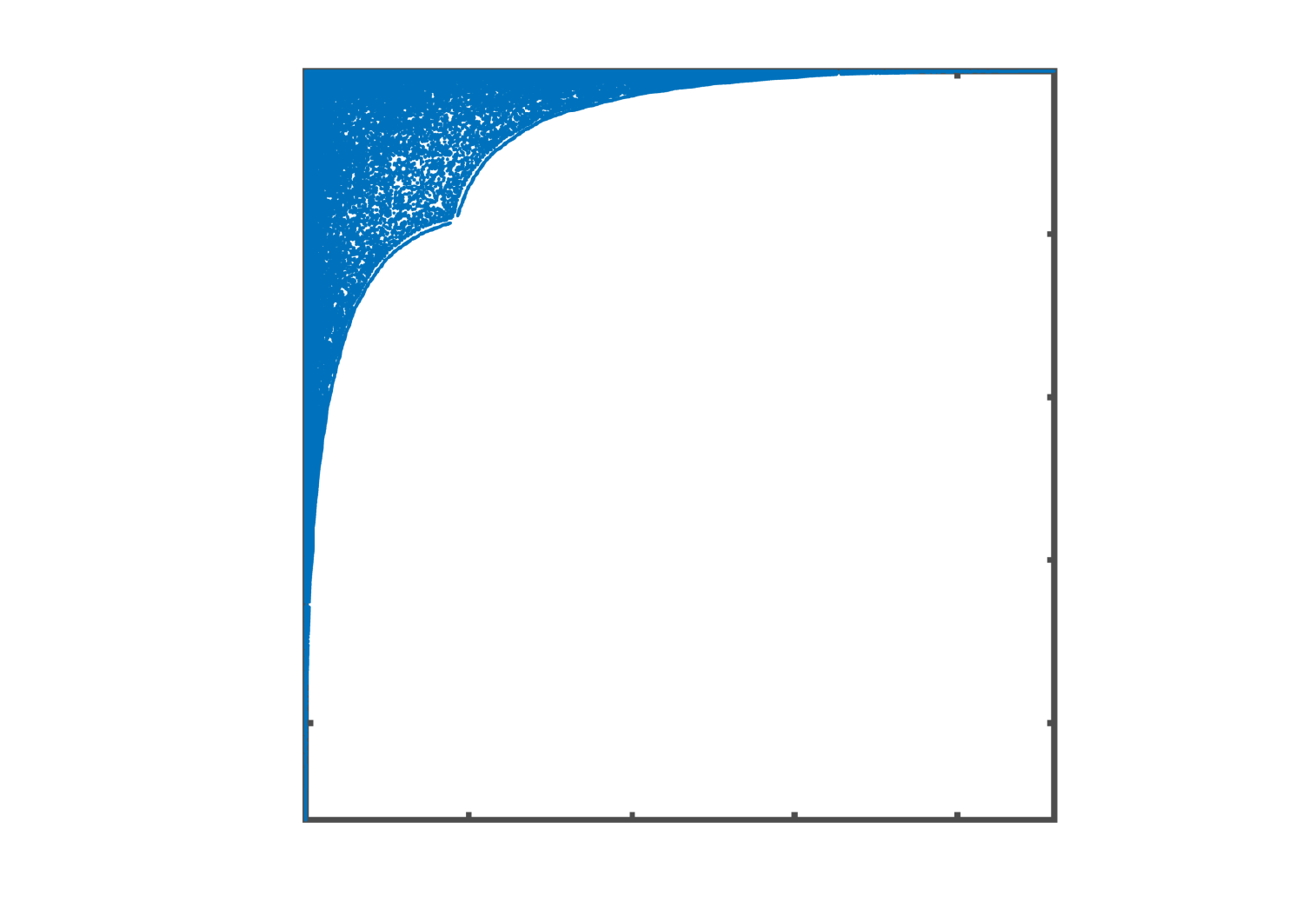}};
\end{tikzpicture}
\caption{Sparsity plots of adjacency matrices for various real-world networks. Each row
of three plots corresponds to a different dataset. Each column corresponds to a different ordering of the network nodes. Left column: nodes ordered by decreasing degree. 
Middle column: nodes ordered by aggregate core score of \cite{rombach2017core}. Right column: nodes ordered by nonlinear spectral method in Algorithm~\ref{alg:1}.}\label{fig:spy_plots}
\end{figure}
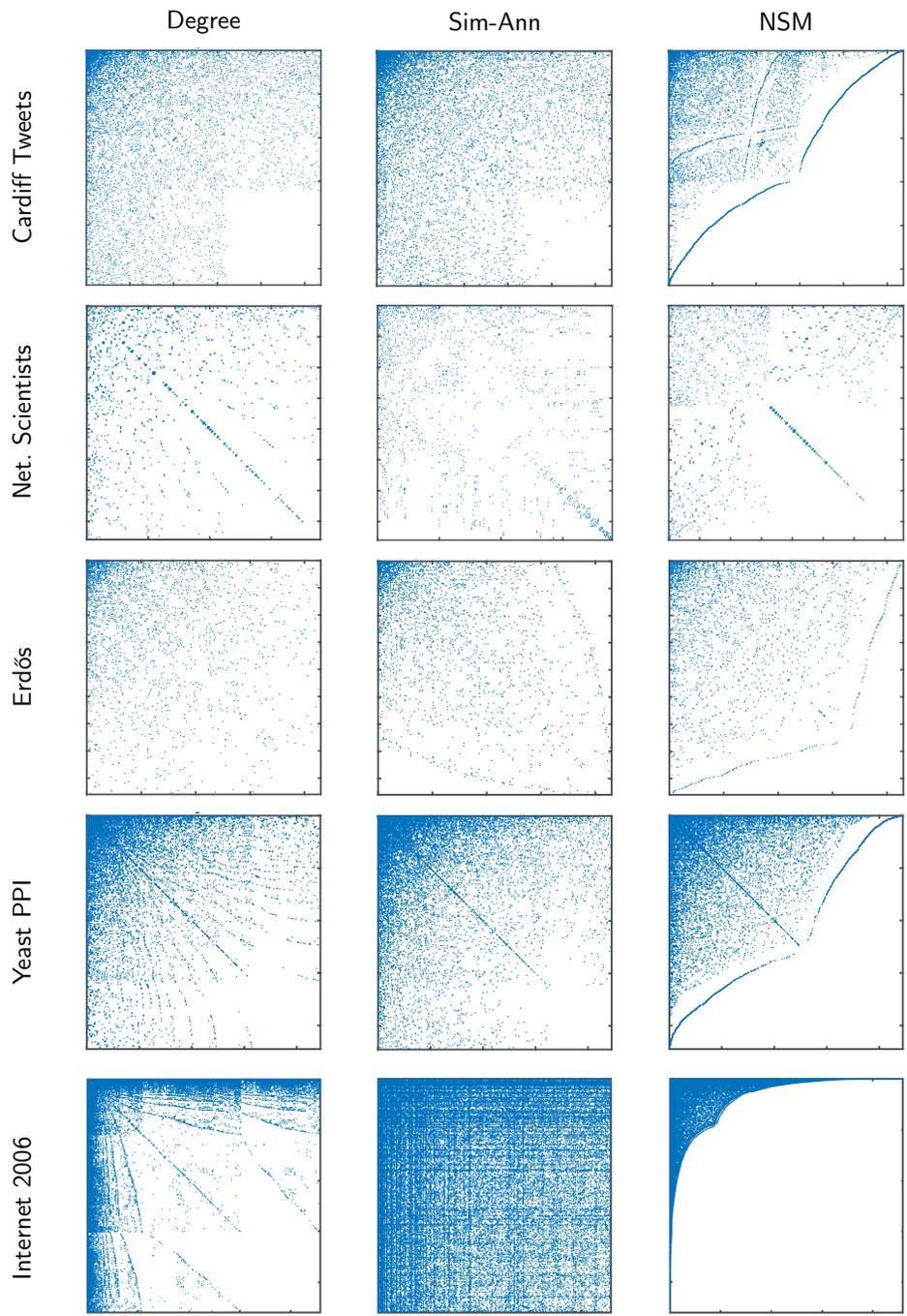

\subsubsection*{Analysis}
\begin{rev}
In Figure~\ref{fig:spy_plots} we use adjacency matrix sparsity plots to show how the three algorithms Degree, Sim-Ann and NSM compare on five networks of different size. 
\end{rev}
In each case, the nodes are reordered in descending magnitude of core--periphery score. 
We see that the three methods give very different visual representations of the data, with 
NSM generally finding a more convincing core--periphery structure.
\begin{rev}
On the Cardiff, 
Erd\H{o}s and Yeast networks, 
NSM gives a well-defined ``anti-diagonal contour'' that essentially separates the
reordered matrix into two regions. This type of behavior 
has been observed for other spectral reordering methods
\cite{TH10}, but does not seem to be fully understood.
\end{rev}

We note that the reciprocated Twitter mentions for the city of Cardiff show a strong 
core--periphery structure in all three orderings. 
Very similar results were observed for all ten city-based networks of reciprocated Twitter mentions collected in \cite{grindrod2016comparison}, which however we refrain from showing here for the sake of brevity. 

\begin{figure}[!t]
\centering
\includegraphics[width=.45\textwidth,clip,trim=0 0 1.5cm 0]{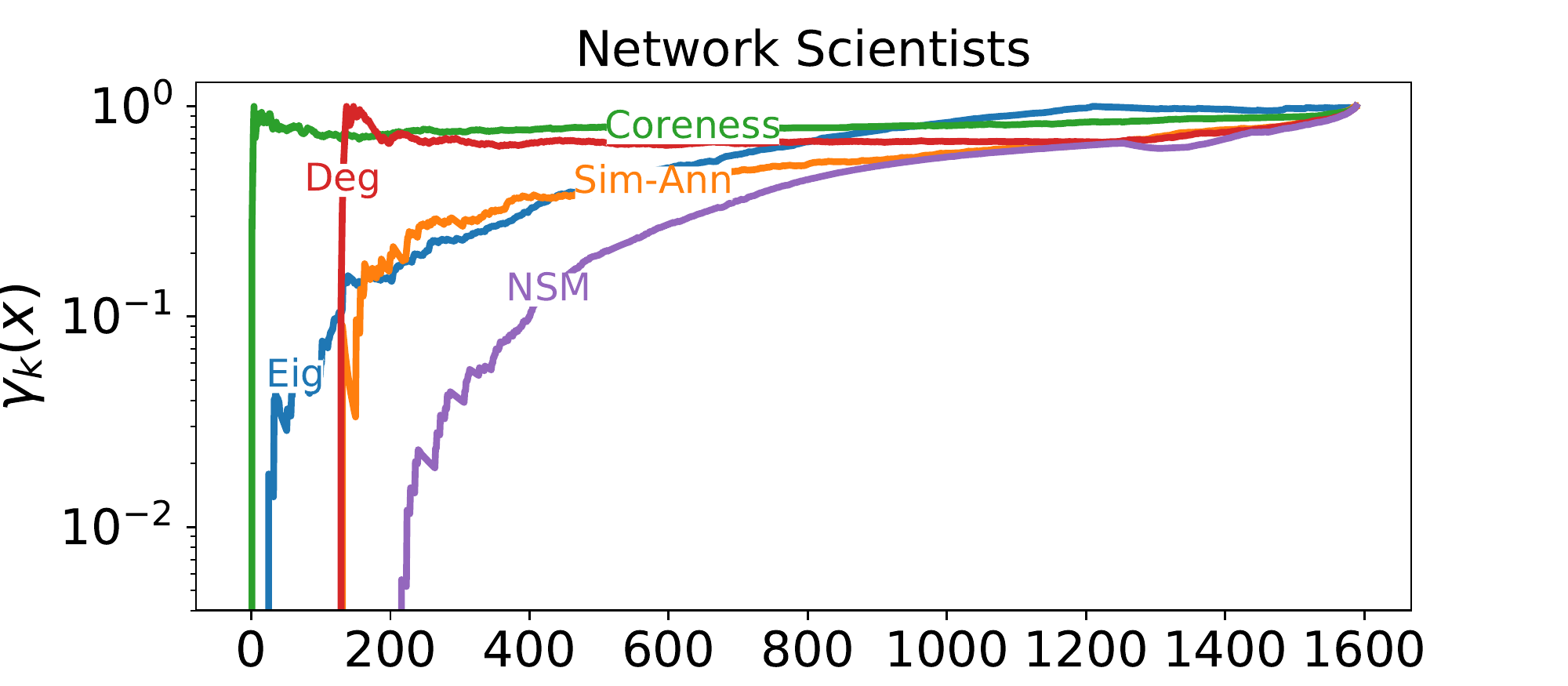}
\includegraphics[width=.45\textwidth,clip,trim=0 0 1.5cm 0]{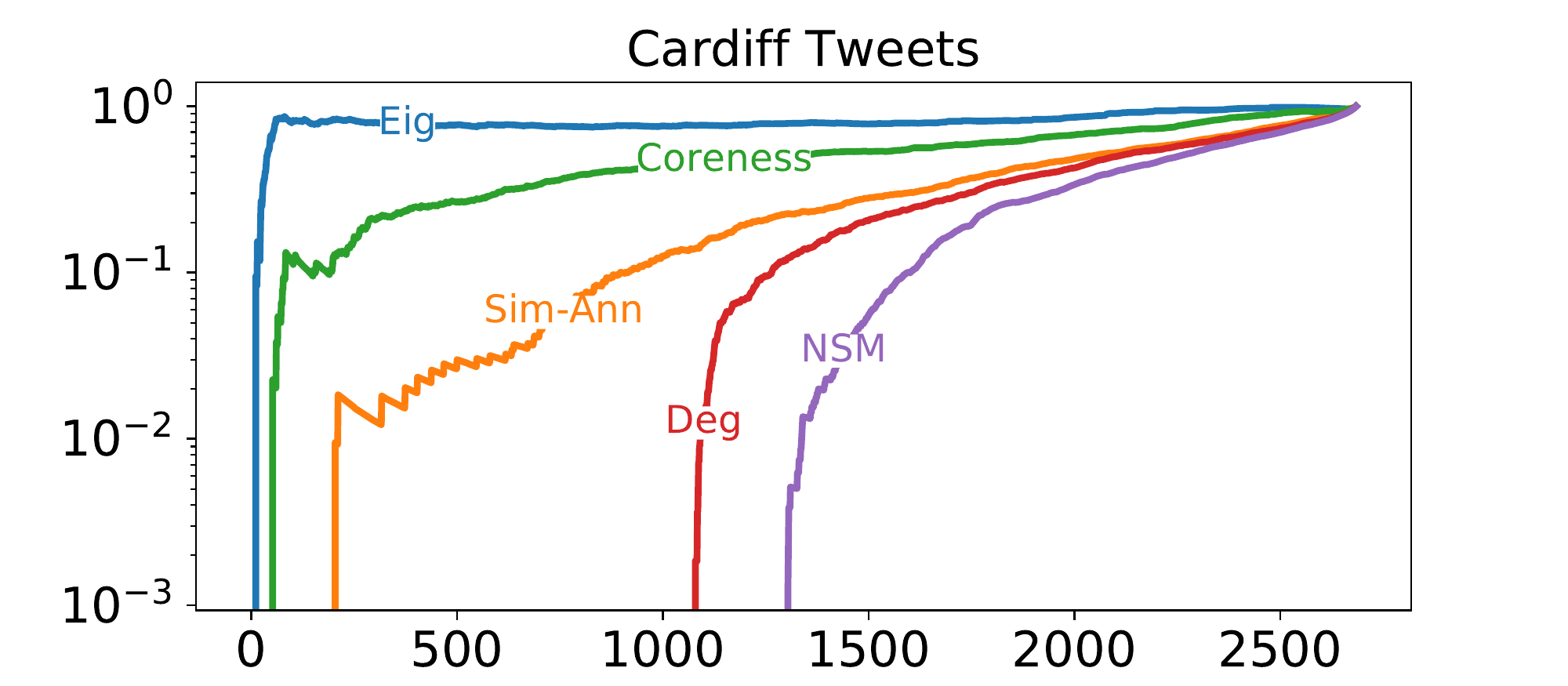}\\
\includegraphics[width=.45\textwidth,clip,trim=0 0 1.5cm 0]{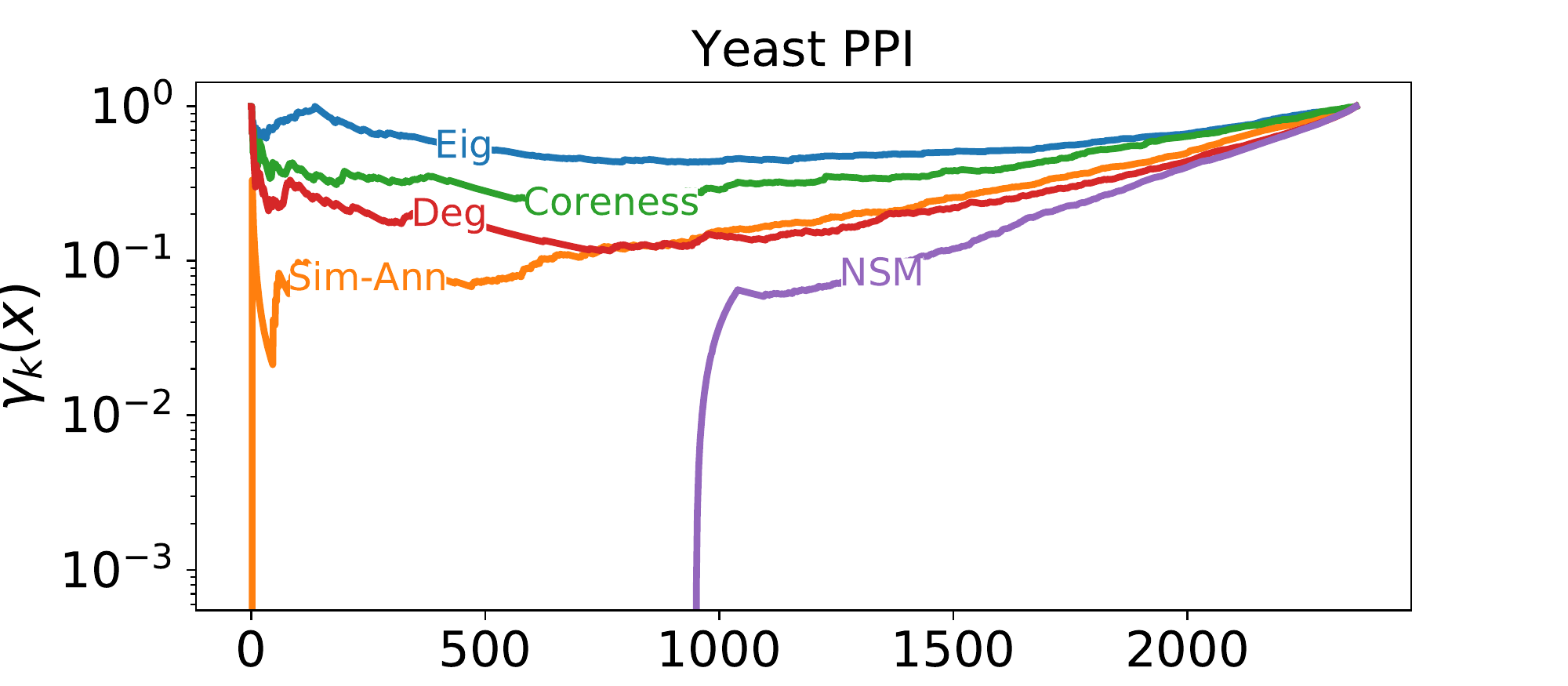}
\includegraphics[width=.45\textwidth,clip,trim=0 0 1.5cm 0]{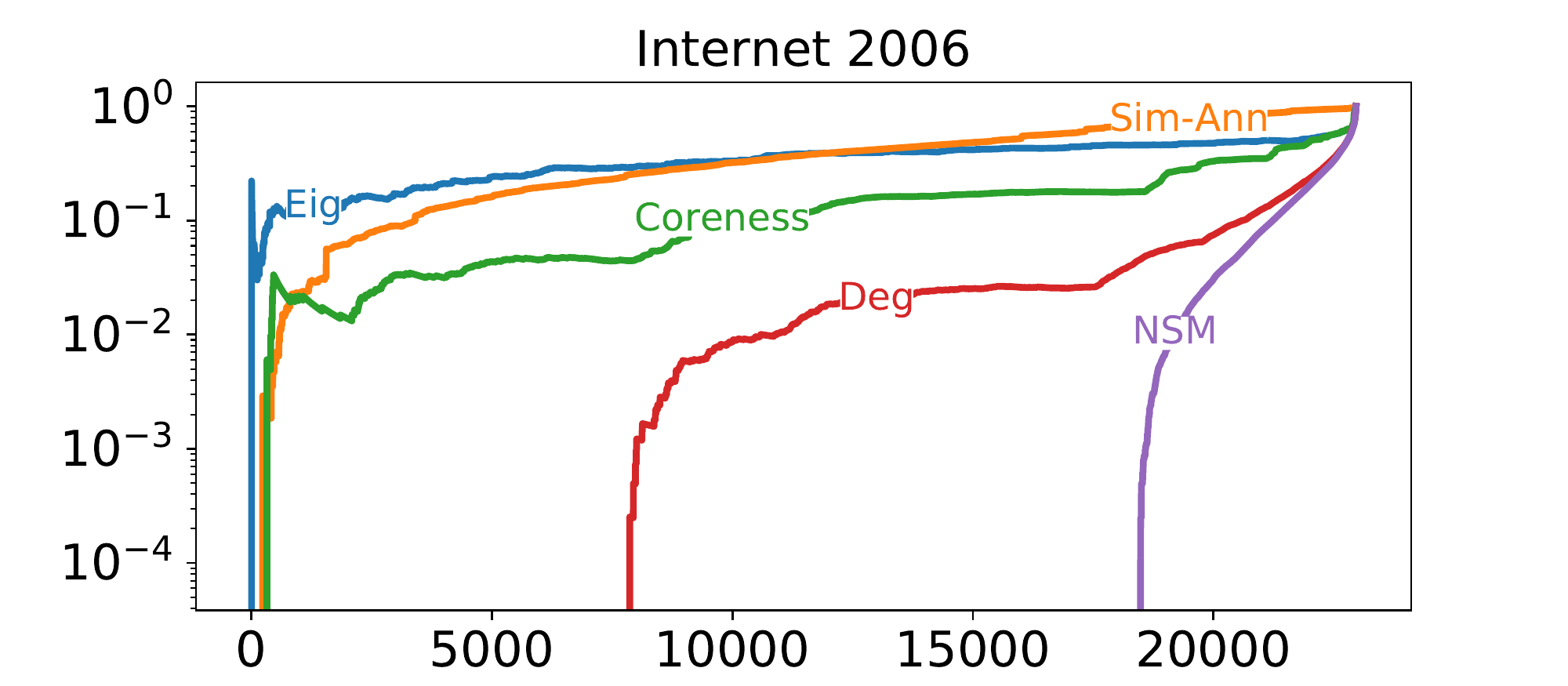}\\
\includegraphics[width=.45\textwidth,clip,trim=0 0 1.5cm 0]{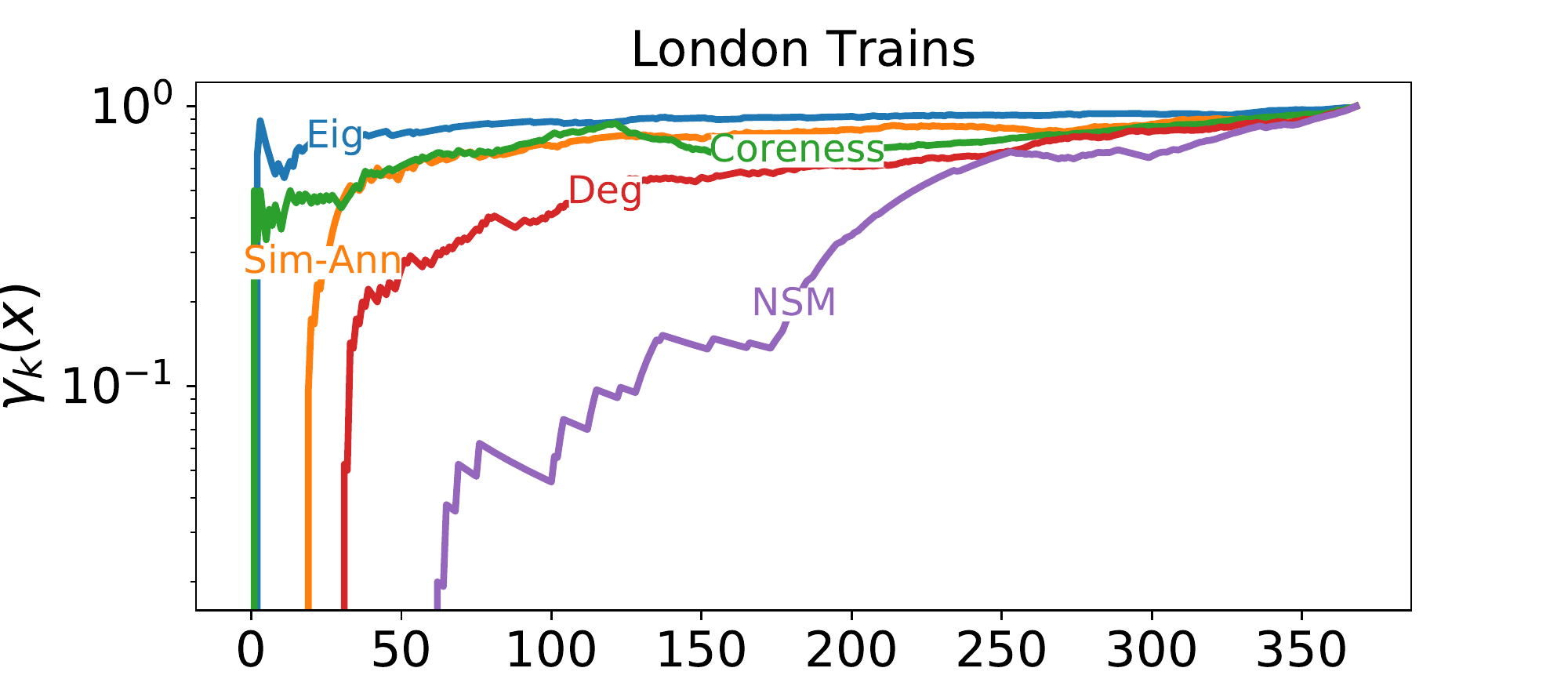}
\includegraphics[width=.45\textwidth,clip,trim=0 0 1.5cm 0]{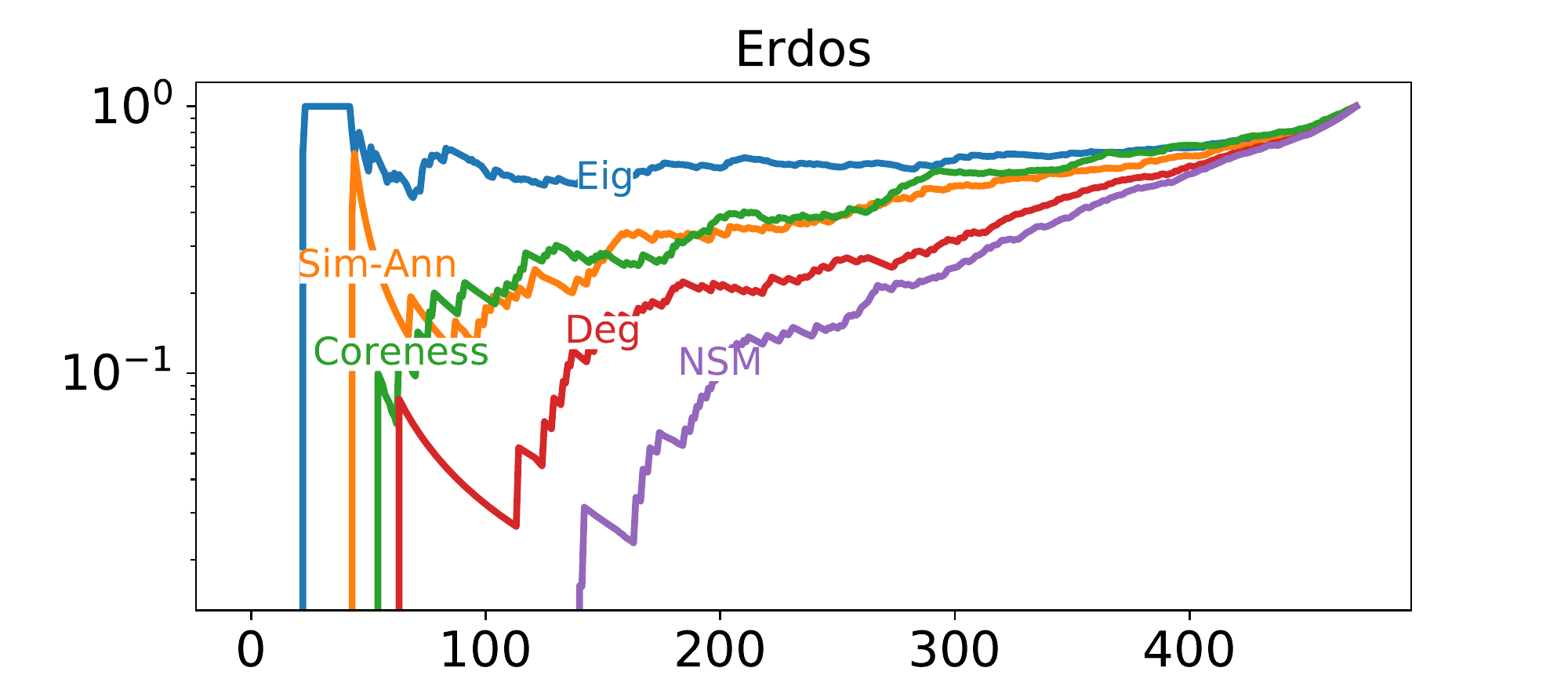}
\caption{(Color online.) Core--periphery profile $\gamma(\b x)$ of three networks where $\b x$ is the core--score vector obtained with five different methods. }\label{fig:cp_profile}
\end{figure}

\begin{rev}
To quantify the quality of the core--periphery assignments and to compare different methods on all the datasets, we perform two further tests. 

In Figure \ref{fig:cp_profile} we show the \emph{core--periphery profile} 
of five networks obtained with different methods. This analysis is inspired by the core--periphery profiling approach proposed in \cite{della2013profiling} and consists of evaluating the core--periphery profile function 
$\gamma(\b x)$ associated with a given core--periphery quality vector $\b x>0$, defined as 
\begin{equation}
\gamma(\b x)_k  = \frac{\sum_{i,j=1}^k A_{\pi_i, \pi_j}}{\sum_{i=1}^k\sum_{j=1}^n A_{\pi_i,j}}, 
\label{eq:cpqv}
\end{equation}
where $\b \pi$ is a permutation such that 
$x_{\pi_1}\leq \cdots \leq x_{\pi_n}$. 
In words, for each $k$ if we regard 
$\pi_1, \pi_2, \ldots, \pi_k$ as peripheral nodes and 
$\pi_{k+1}, \pi_{k+2}, \ldots, \pi_n$ as core nodes then
$\gamma(\b x)_k$ in (\ref{eq:cpqv}) measures the ratio of 
periphery-periphery links to
periphery-all links.
Hence, $\b x$ reveals a strong core-periphery structure 
if $\gamma(\b x)_k$ remains small for large $k$.

The quantity $\gamma(\b x)_k$ also has an interesting random walk interpretation. Given $\b x>0$ let  $S_k = \{\pi_1,\dots,\pi_k\}$ and consider the standard random walk on $G$ with transition matrix $T=(t_{ij})$ defined by $t_{ij} = a_{ij}/\sum_k a_{ik}$. As the graph is undirected, the stationary distribution of the chain $\b y>0$ is the (normalized) degree vector $\b y = \b d/\sum_i d_i$. Therefore, 
$$\gamma(\b x)_k = \frac{\sum_{i,j\in S_k}y_i t_{ij}}{\sum_{i\in S_k}y_i}\, ,$$ which corresponds to the persistance probability of $S_k$, 
i.e., the probability that a random walker who is currently in any 
of the nodes of $S_k$, remains in $S_k$ at the next time step. 
Clearly $\gamma(\b x)_k\leq \gamma(\b x)_{h}$ if $k\leq h$ and 
$\gamma(\b x)_n = 1$, for any $\b x$. 
This further justifies why having small values of $\gamma(\b x)_k$ for 
large values of $k$ is a good indication of the presence of a core and periphery \cite{della2013profiling}. Figure \ref{fig:cp_profile} shows that the smallest core--periphery profile $\gamma(\b x)$ is obtained when $\b x$ is the output of 
Algorithm~\ref{alg:1}. This confirms the behavior shown in 
Figure~\ref{fig:spy_plots}---Algorithm~\ref{alg:1} is the most effective 
at transforming each network into core-periphery form.
\end{rev}

\begin{figure}[t!]
\centering
\includegraphics[width=.9\textwidth,clip,trim=0cm 0cm 0cm 0]{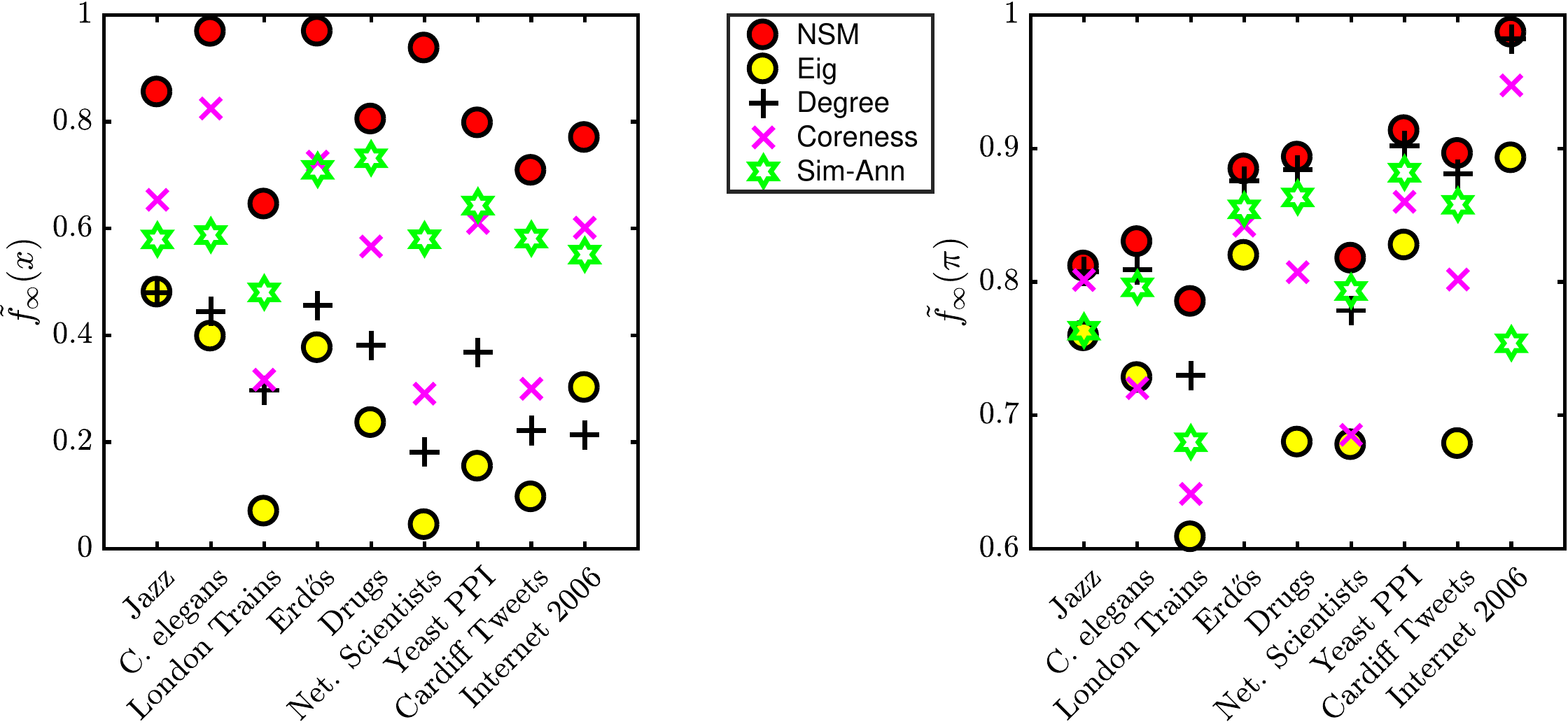}
\caption{(Color online.) Normalized core--periphery quality measure $\tilde f_\infty$ for different methods. Left: Value of $\tilde f_\infty(\b x)$ for different core--score vectors $\b x$ normalized so that $\max(\b x) = 1$ and $\min(\b x) = 0$. Right: Value of $\tilde f_\infty(\b \pi)$ where $\b \pi$ is the permutation  that sorts the entries of $\b x$ in increasing order.}\label{fig:cp_quality}
\end{figure}

\begin{rev}
Finally, in Figure~\ref{fig:cp_quality} we compare the value of the 
core--periphery quality function $f_\infty$ on all the 
datasets and all the methods. 
To cover networks of different sizes we plot the normalized value 
$$
\tilde f_\infty(\b x) = \frac{f_\infty(\b x)}{(\max_i x_i)\sum_{ij}a_{ij} }\, .
$$ 
Precisely, the figure shows two plots: On the left we evaluate $\tilde f_\infty(\b x)$ on core--score vectors $\b x$ obtained by the methods, rescaled so that $\max(\b x) = 1$ and $\min(\b x) = 0$, whereas on the right we evaluate $\tilde f_\infty$ on  the corresponding permutation vector $\b \pi$ such that $x_{\pi_1}\leq \dots\leq x_{\pi_n}$. The NSM is designed to optimize $f_\alpha$ (recall  $\alpha = 10$ in our experiments) 
so the value of $\tilde f_\infty$ is significantly larger than the 
value obtained with other methods. 
We see that NSM continues to give the best results 
when $\tilde f_\infty$ is evaluated on the 
associated permutation.
\end{rev}

\begin{figure}[!t]
\centering
\includegraphics[width=.63\textwidth,clip,trim=1.5cm 3.5cm 1cm 3.5cm]{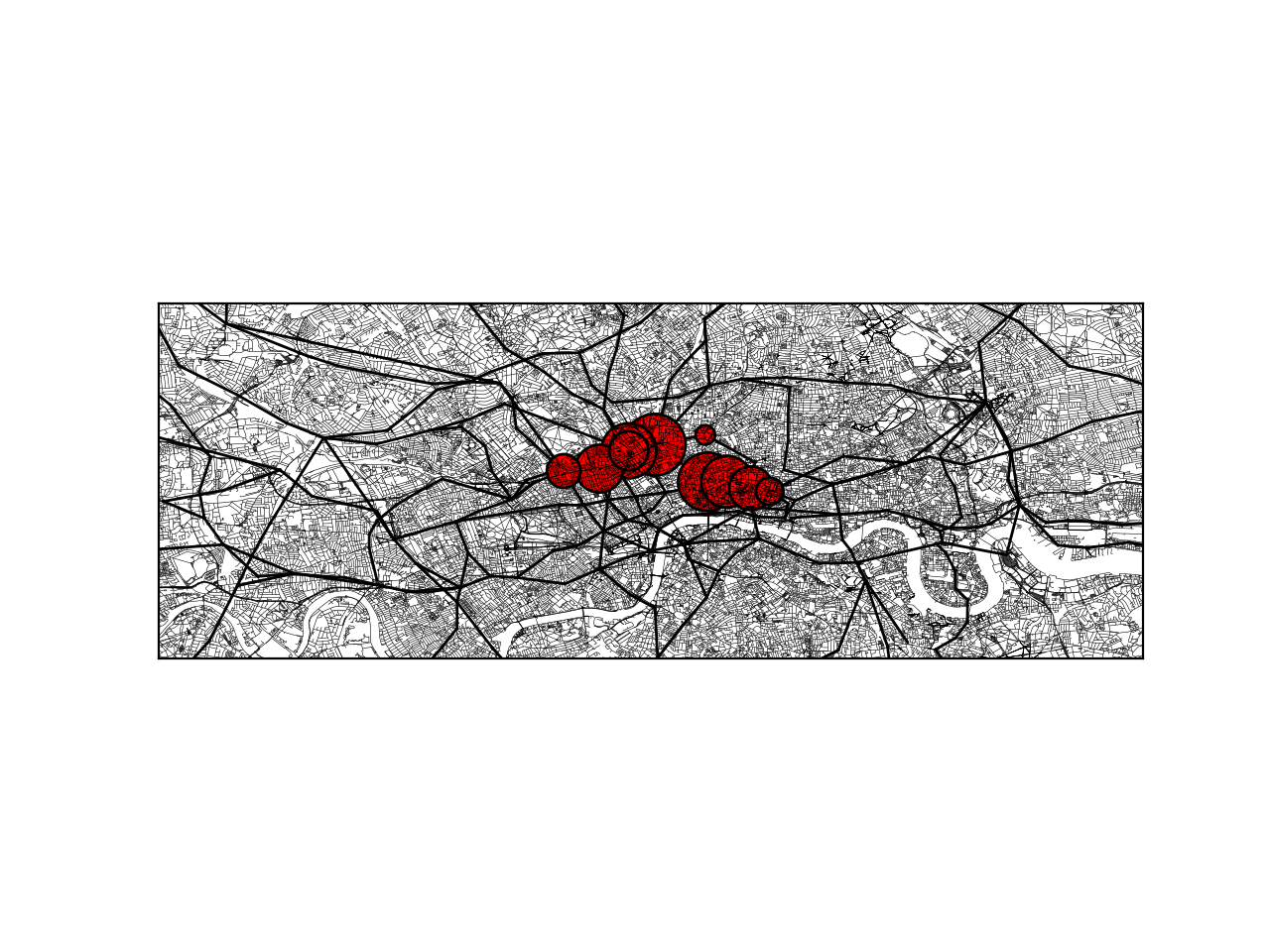}\\
\includegraphics[width=.63\textwidth,clip,trim=1.5cm 3.5cm 1cm 3.5cm]{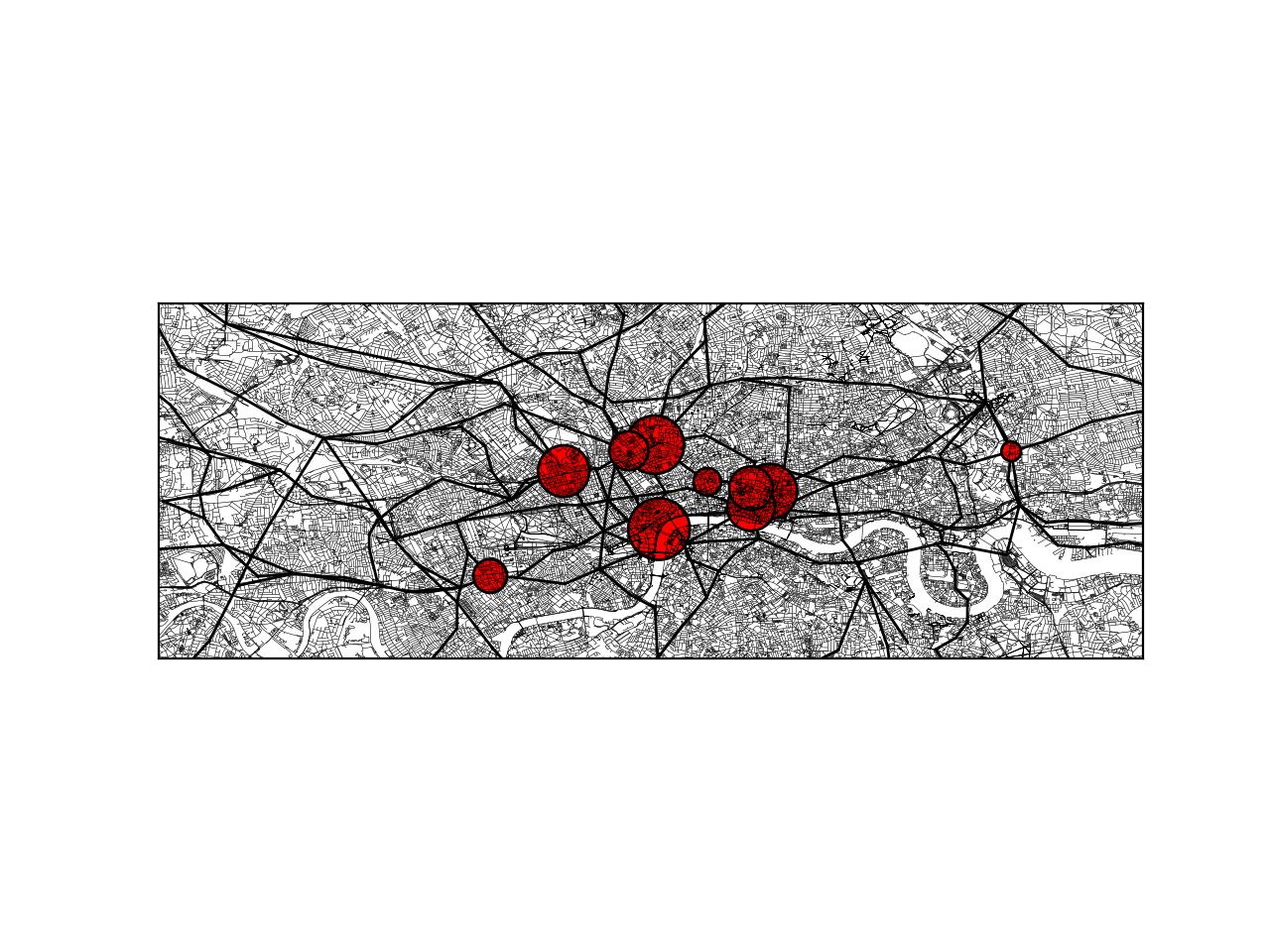}\\
\includegraphics[width=.63\textwidth,clip,trim=1.5cm 3.5cm 1cm 3.5cm]{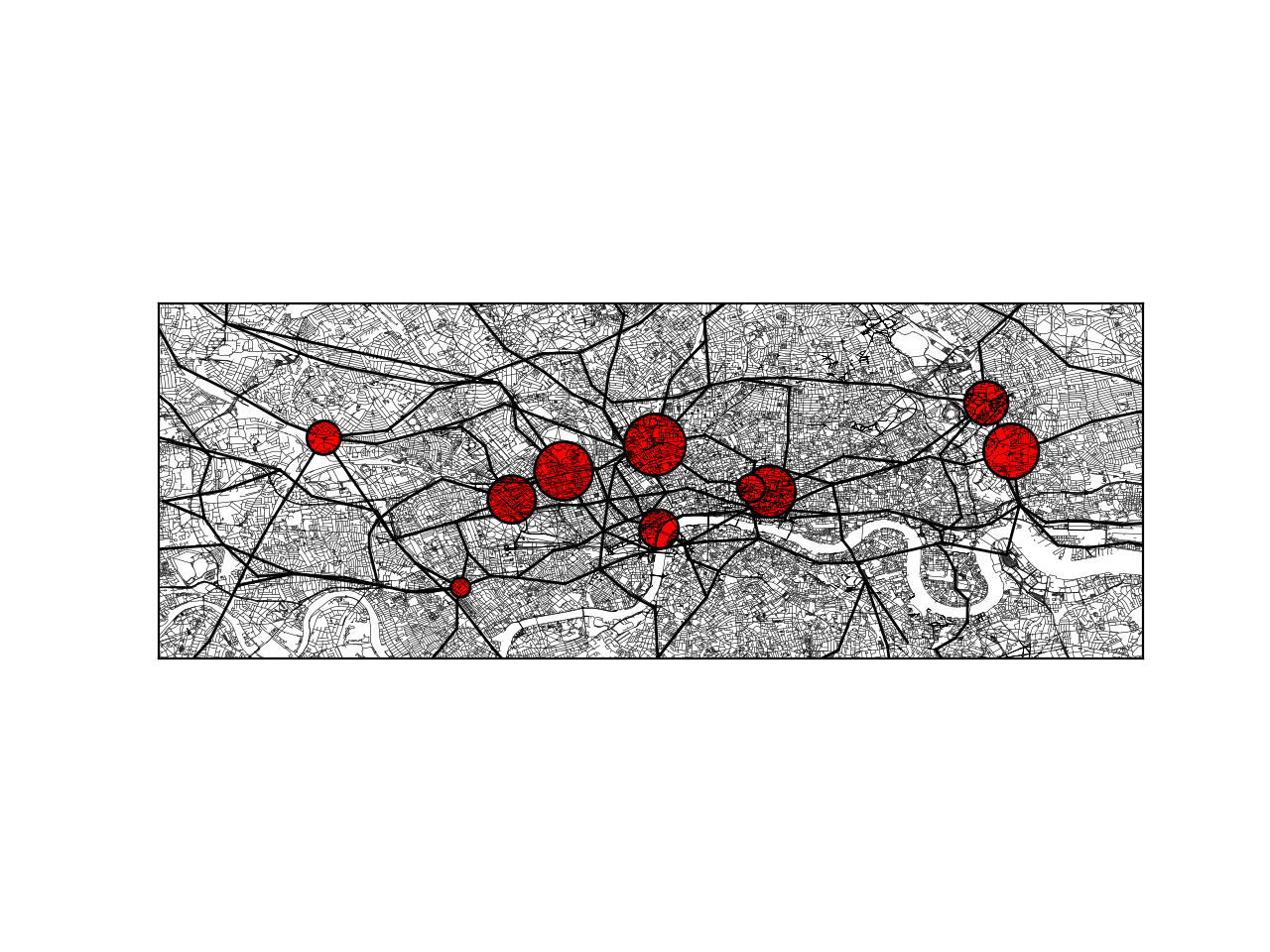}
\caption{(Color online.) Physical layout of the aggregate London transportation network. Red circles indicate the ten nodes with highest core--periphery score for the three algorithms, with node size
proportional to score. 
From top to bottom: Eigenvector centrality, Sim-Ann and NSM.}\label{fig:drawing_london}
\end{figure}

\subsubsection*{London Transportation Network}
\begin{rev}
In a final experiment we look in more detail at the 
London transportation network, where further nodal information is 
available,
using the Perron--Frobenius eigenvector of the adjacency matrix as a baseline for comparison. As discussed in Section \ref{sec:methods}, this vector can be viewed as both a
centrality measure and a core--periphery score, and it corresponds to a linear counterpart of our approach, retrieved when $\alpha\to 0$.  
We compare central nodes in the London train network obtained 
from Eig, NSM and Sim-Ann.
Note that important nodes for both Eig and Sim-Ann are somewhat related with 
the concept of centrality, as both methods aim to maximize the 
same core--quality function $\sum_{ij}A_{ij}x_ix_j$ but force 
different constraints sets, $\mathit \Omega = \{\b x: \|\b x\|_2=1\}$ for Eig and $\mathit \Omega = C_{\alpha, \beta}$ for Sim-Ann. 
On the London train network we find that the core assignments of these two techniques highly correlate. The importance of nodes captured by the NSM is, instead, more directly related with core  and periphery features and significantly differ from  Eig and Sim-Ann. For the sake of brevity we do not compare with other methods here. 
\end{rev}

In Figure~\ref{fig:drawing_london} 
we display the edges (underground, overground and DLR connections)
in physical space, 
with darker linetype indicating a larger weight. 
The top ten stations are highlighted for the three measures, with node size
proportional to the value.
Although four stations are highlighted in all three plots, 
there are clear differences in the results.
Eigenvector centrality and Sim-Ann produce similar results, 
focusing on a set of stations that are geographically close, 
whereas 
NSM assigns higher core scores to some stations at 
key intersections that are further from the city centre. 

To underscore the differences that are apparent in 
Figure~\ref{fig:drawing_london}, 
in 
Table~\ref{tab:data_london} we list the names of the top ten stations drawn in Figure~\ref{fig:drawing_london},
for each of the three rankings. 
Whereas four major stations, namely Baker Street,  King's Cross, Liverpool Street and Moorgate, are shared by 
all three methods, four stations appearing in the 
NSM top ten do not appear in the top ten of the other two methods.
Table~\ref{tab:data_london} 
also gives 
the overall number of passengers entering or exiting each 
station.
A station may play more than one role 
(underground, overground or DLR)
and we list the most recently reported total annual usage.
More precisely, we sum the records for 
\begin{itemize}
\item London Underground annual entry and exit 2016, 
\item National Rail annual entry and exit, 2016--2017,
\item DLR annual boardings and alightings, 2016,
\end{itemize}
as reported in Wikipedia in April 2018.  
Numbers indicate millions of passengers. The last row shows the overall number of passengers using the top ten stations identified by each method. We note that none of the rankings orders the
stations strictly by passenger usage. 
However, while the top ten stations selected by both Eigenvector and Sim-Ann involve around 600 million passengers per year, 
the top ten NSM stations 
involve almost 800 million passengers.

\begin{table}[!t]
\centering
\sf \footnotesize
\begin{tabular}{|llllll|}
\hline
\textbf{Eigenvector}  		& &  \textbf{Sim-Ann}  & & \textbf{NSM} & \\
\hline
King's Cross     &128.85	& Embankment	  &26.84	& King's Cross     &128.85	\\
Farringdon       &29.75    	& King's Cross	  &128.85	& Baker Str.       &29.75 	\\
Euston Square    &14.40 	& Liverpool  Str. &138.95   & West Ham         &77.10 	\\
Barbican         &11.97     & Baker Str.      &29.75    & Liverpool Str.   &138.95	\\
Gt Port.\ Str. &86.60   	& Bank            &96.52    & Paddington       &85.32   \\
Moorgate         &38.40     & Moorgate        &38.40    & Stratford        &129.01 	\\
Euston           &87.16     & Euston Square   &14.40    & Embankment       &26.84 	\\
Baker Str.	     &29.75     & Gloucester Road &13.98    & Willesden Junct. &109.27	\\
Liverpool Str.   &138.95 	& Farringdon      &27.92    & Moorgate         &38.40   \\
Angel            &22.10     & West Ham        &77.10    & Earls Court 	   &20.00	\\
\hline
\hline
\textit{\textbf{Total}} & 586.09 & & 592.70 & & 783.48\\
\hline
\end{tabular}
\caption{Ten London train stations with highest core value, 
according to eigenvector centrality (left column), Sim-Ann (middle column), and NSM (right column), applied to the weighted London trains network. The numbers beside each station show overall 
(underground, overground, DLR) annual usage in millions of passengers. Numbers in the bottom row show the sum of annual usage across the top 10 stations selected by each method.
King's Cross refers to a combination of 
King's Cross and St Pancras main-line stations and the 
King's Cross St Pancras underground station.
}\label{tab:data_london}
\end{table}
\begin{figure}[!t]
\begin{tabular}{lrrr}
 & \includegraphics[width=.27\textwidth]{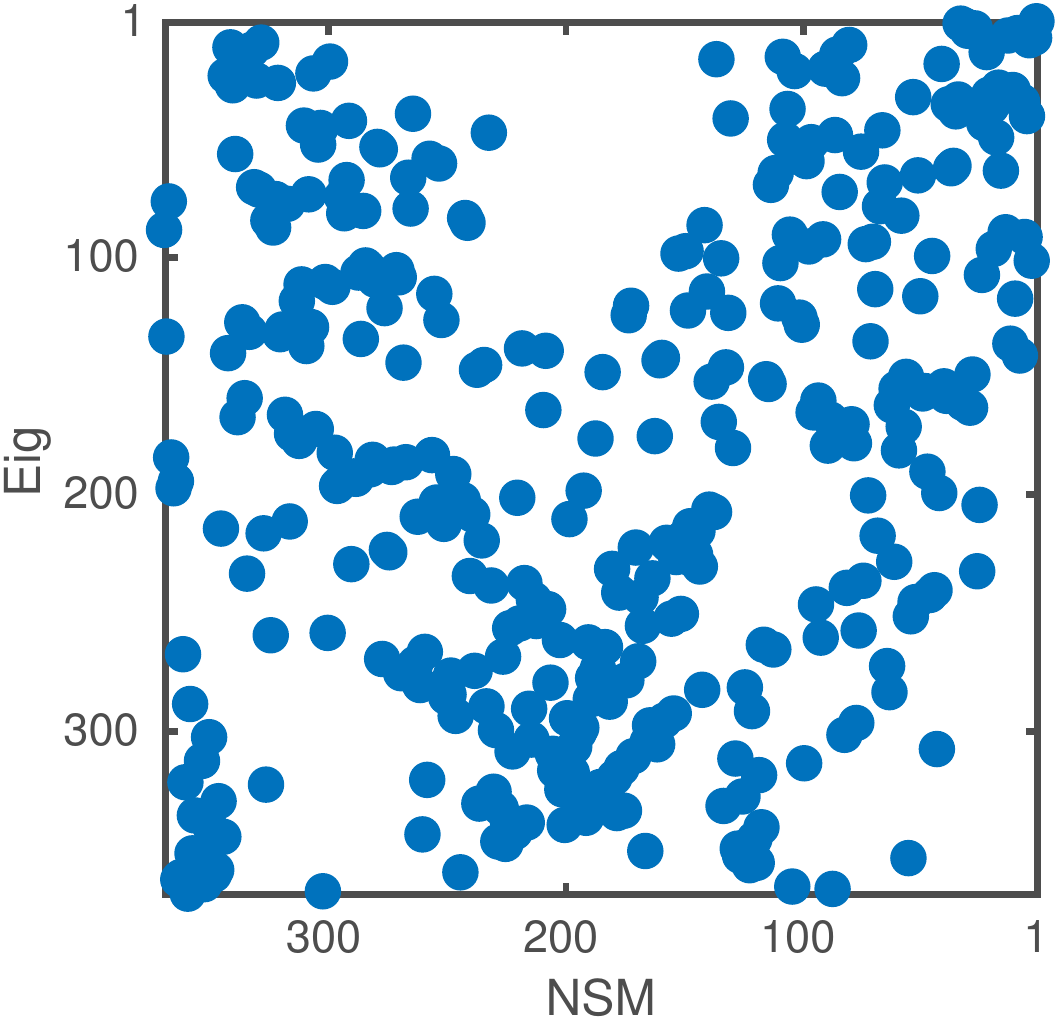}
 & \includegraphics[width=.27\textwidth]{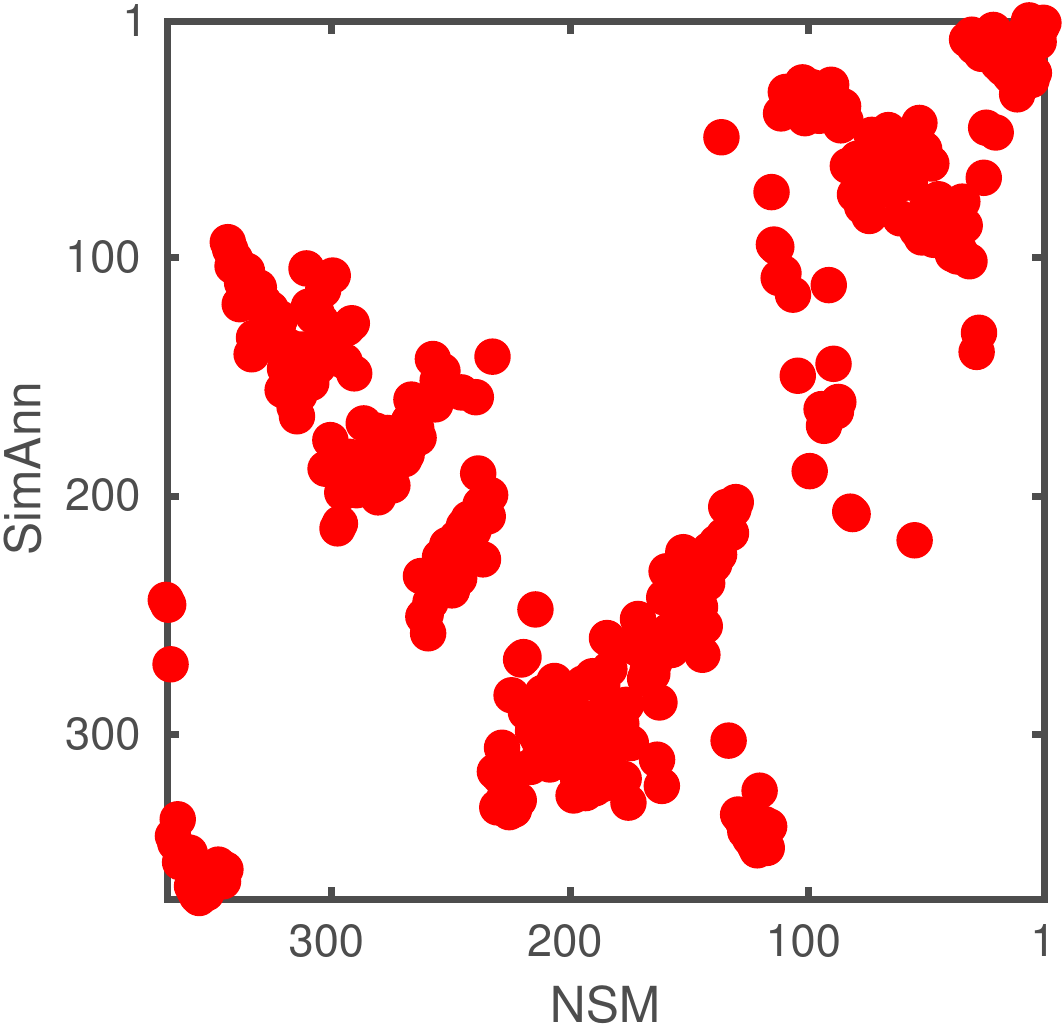}
 & \includegraphics[width=.27\textwidth]{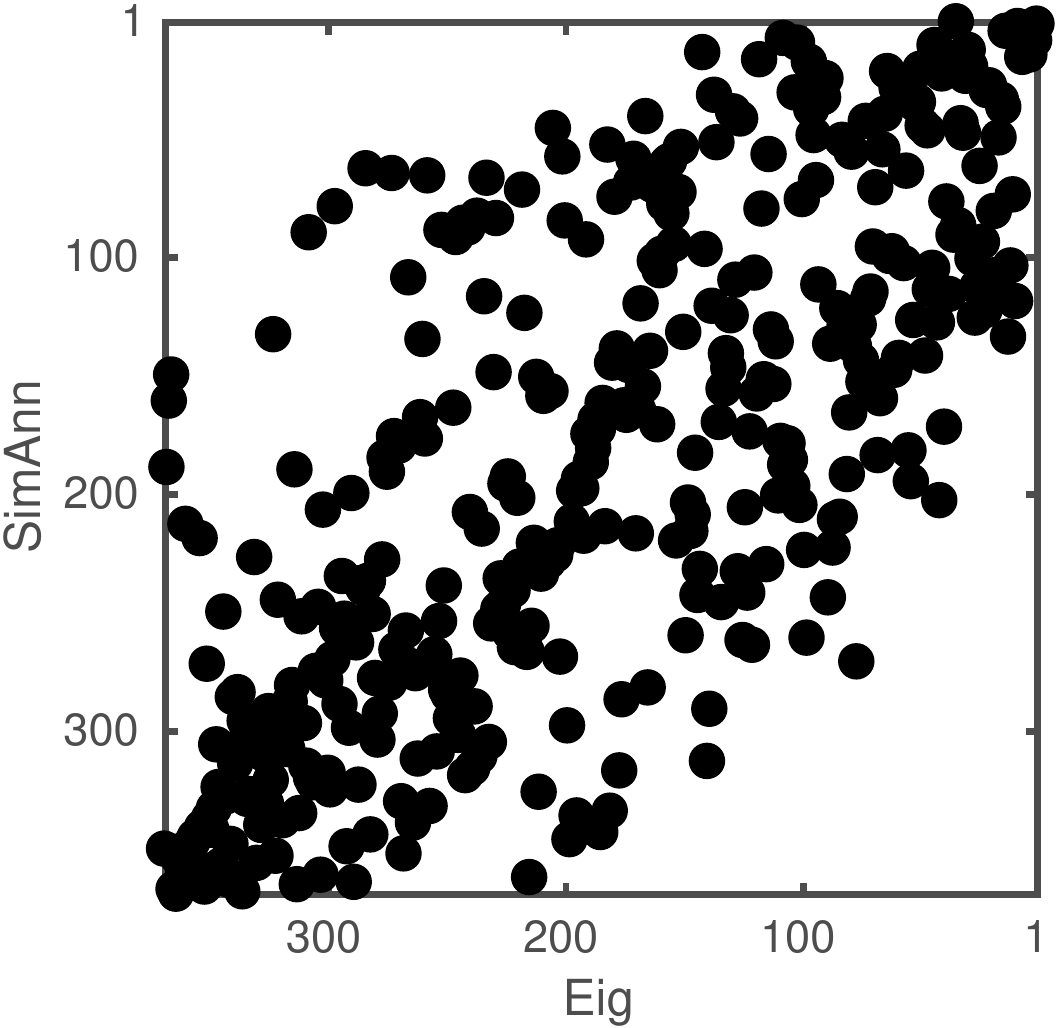} \\
 \hline
 \textit{Kendal $\tau$} \!\!\!\!\!\!\!\!\!\!\!\!\!& 0.1442 & 0.2930 & 0.5455
\end{tabular}
\caption{Top: Scatter plots comparing the ranking associated with the three core score functions: Eigenvector centrality, Sim-Ann and NSM. Bottom: Kendal $\tau$ correlation coefficient between the three pairs of rankings shown in the corresponding scatter plot.
}\label{fig:scatter_london}
\end{figure}

For a comparison across all $369$ stations,  
Figure~\ref{fig:scatter_london} scatter plots the 
rankings for the three methods in a pairwise manner. 
We see that the left and middle plots, 
NSM versus Eig and NSM versus Sim-Ann, 
show much less agreement 
than the third, Eig versus Sim-Ann. This is confirmed by the Kendal $\tau$ correlation coefficients between the different rankings, shown at the bottom of  Figure~\ref{fig:scatter_london}.

\section{Discussion}\label{sec:disc}
The approach in \cite{borgatti2000models,rombach2014core,Zhang15}
was to set up a discrete optimization problem and then apply heuristic algorithms that are not 
guaranteed to find a global minimum.
Our work differs by 
relaxing the problem before addressing the 
computational task.
We showed that a relaxed analogue of a natural 
discrete optimization problem 
allows for a globally convergent iteration that 
is feasible for large, sparse, networks.
This philosophy is in line with classical
and widely used reordering and clustering methods that make use of the 
Fiedler or the Perron--Frobenius eigenvectors \cite{EH10}. However,  in the core--periphery setting 
considered here, the resulting relaxed problem is equivalent to an eigenvalue problem that is inherently nonlinear and is reminiscent of more recent clustering and reordering techniques that exploit nonlinear eigenvectors \cite{buhler2009spectral,tudisco2017node,tudisco2018nodal,tudisco2017community}.
Hence, we developed new results in nonlinear Perron--Frobenius theory 
in order to derive and analyze the algorithm. 

As with all clustering, partitioning and reordering methods in network science, there is no 
absolute gold standard against which to judge results---the underlying problems may be defined
in many different ways.
In this work we introduced a new random graph model that 
(a) gives further justification for our algorithm, and 
(b) provides one basis for systematic comparison of methods.
Maximum likelihood results on synthetic networks with 
planted structure 
showed the effectiveness of the new method, as did 
qualitative visualizations 
and 
quantitative tests  
across a range of application areas.

\bigskip

\textbf{Acknowledgments}
We thank Mason A. Porter  
for supplying code 
that implements the algorithm in 
\cite{rombach2014core,rombach2017core}. 
EPSRC Data statement: 
all data and code related to this work are 
publicly available, 
and may be obtained by following the 
links in the text or by consulting the associated references.


\end{document}